\documentclass[orivec]{article}
\usepackage{amsmath}
\usepackage{amsfonts}
\usepackage{amssymb}
\usepackage{amsthm}
\usepackage{graphicx}
\usepackage{stmaryrd}
\usepackage{color}
\usepackage{MnSymbol}
\usepackage{breakcites}

\newcommand{\cn}{\mathrm{Cn}}
\newcommand{\conp}[1]{\mathrm{Cn}(\{ {#1}\})}
\newcommand{\cons}[1]{\mathrm{Cn}(#1)}
\newcommand{\mb}[1]{\mathbb{#1}}

\newtheorem{DEF} {Definition}
\newtheorem{OBS} {Observation}
\newtheorem{THE} {Theorem}
\newtheorem{LEM} {Lemma}

\newcommand{\cro}{\ast_c}

\newcommand{\cco}{\div_c}

\newcommand{\pco}{\div_p}

\newcommand{\meo}{\dotplus}

\newcommand{\Ap}{A^{\prime}}
\newcommand{\Bp}{B^{\prime}}
\newcommand{\Kone}{K_1}
\newcommand{\Ktwo}{K_2}
\newcommand{\MU}{\textrm{$*$}^{\!\!\!\!\textrm{\small{$\sim$}}} }
\newcommand{\gammap}{\gamma^{\prime}}

\newcommand{\phinegphi}{\{\varphi, \neg \varphi\}}
\newcommand{\psinegpsi}{\{\psi, \neg \psi\}}

\newcommand{\taut}{\textrm{\large{$\top$}}}  
\newcommand{\falsum}{\textrm{\large{$\bot$}}}

\begin{document}

\title{\textbf{Choice revision on belief bases}}
\author{
\small{Li Zhang}
}
\date{\small{KTH Royal Institute of Technology}}

\maketitle

\begin{abstract}
\footnotesize
\noindent
In this contribution we explore choice revision, a sort of belief change in which the new information is represented by a set of sentences and the agent could accept some of the sentences while rejecting the others. We propose a generalized version of expansion operation called partial expansion for developing models of choice revision.  By using the partial expansion and two multiple contraction operations previously introduced in the literature, we construct two kinds of choice revision on belief bases. For each of them we propose a set of postulates and prove a partial or full representation theorem. Furthermore, we investigate the operations of making up one's mind derived from these two kinds of choice revision and also give the associated representation theorems.  
\end{abstract}

\section{Introduction}

Choice revision\footnote{\,This term is firstly introduced by \cite{Fuhrmann_phd}. To refer to the same concept, \cite{rott_change_2001} uses the terms ``bunch revision'' and ``pick revision'', while \cite{falappa_prioritized_2012} use the term ``selective change''.} is a sort of  non-prioritized multiple belief revision. It has two particular characteristics which make it distinctive from the standard AGM revision \cite{alchourron_logic_1985}: \textit{multiple} which means that in this belief change the new information is represented by a set of sentences instead of a single sentence, and \textit{non-prioritized} which means that the new information has no priority to the original beliefs, in other words, the agent could accept part of the new information while rejecting the rest. As suggested by \cite{falappa_prioritized_2012}, a practical scenario of choice revision could be a multi-agent system in which a computational agent receives independent information from other agents and adds the most reliable part to her knowledge base. 

Choice revsion is differrent from the ``selective revision'' introduced in Fermé  Hansson (1999). Selection revision is a also kind of non-prioritized revision, but its inputs are sentences rather than sets of sentences. It is easy to see that generally choice revision by a  finite set $A$ cannot be reduced to selective revision by the conjunction $\& A$ of all elements in $A$: Selection revision $\ast_s$ is postulated to satisfy extensionality, i.e. if $\varphi$ is logically equivalent to $\psi$, then that $K \ast_s \varphi = K \ast_s \psi$ holds. However, though $\varphi \wedge \neg \varphi$ is equivalent to $\psi \wedge \neg \psi$ for every sentences $\varphi$ and $\psi$, it is possible that $K \cro \{\varphi, \neg \varphi\} \neq K \cro \{\psi, \neg \psi\}$, for example, in the case of that $\varphi$ and $\psi$ are stating totally irrelevant things.  

Choice revision is also different from the operation of \textit{merge}  \cite{Fuhrmann1996_book,falappa_prioritized_2012}, which is another mode of non-prioritized multiple revision. In merge, the original beliefs and the new information are symmetrically treated. So, it holds for any merge operator $\#$ that $K \# A = A \# K$ for every sets $K$ and $A$. It follows that $\#$ does not in general satisfy the following success condition: $A \cap (K \# A) \neq \emptyset$ when $A$ is non-empty. Otherwise, it would lead to an unacceptable consequence that $\{\varphi\} \# \{\neg \varphi\}  = \{\varphi , \neg \varphi\} $ for every $\varphi$. However, as we will show, choice revision could satisfy the above success condition without resulting such unreasonable consequence.

Little work has been done on the formal properties of choice revision. An initial study on this issue could be found in \cite{ 2018arXiv180407602Z}, where a newly developed framework of belief change called \textit{descriptor revision} \cite{hansson_descriptor_2013} played an important role. As shown in \cite{ 2018arXiv180407602Z}, it is difficult to model choice revision in the traditional AGM framework since the select-and-intersect method employed in this framework is not generally applicable. But instead the descriptor revision which employs a select-direct methodology is workable. This is also the reason why \cite{zhang_how_2015} used this new framework to develop the modelling for the operation of making up one's mind, which is a belief change operation that takes the agent from a state in which she neither believes in a sentence nor in its negation to a state in which she believes in one of them. In the formal sense, choice revision is a generalization of the operation of making up one's mind, as the operator $\MU$ of making up one's mind can be reconstructed from choice revision operator $\cro$ in the way of $K \MU \varphi = K \cro \{\varphi, \neg \varphi\}$.  

Unfortunately, though descriptor revision is a so powerful tool (see \cite{hansson_descriptor_2017} for a survey), it seems that the framework is only suitable for modelling belief changes on \textit{belief sets}, i.e. sets of sentences that are logically closed. To see this, a quick review on the basics of descriptor revision is necessary.

In the framework of descriptor revision, a special kind of metalinguistic sentences called \textit{descriptors} are introduced for describing the success conditions of belief changes. Descriptor revision operator $\circ$ is a unified operator which applies to all descriptors. More specifically, an \textit{atomic} descriptor is a sentence $\mathfrak{B}\varphi$ such that $\varphi$ is sentence from the object language and $\mathfrak{B}$ is a metalinguistic operator. A \textit{composite} descriptor (descriptor for short) is a set of truth-functional combinations of atomic descriptors. The intended interpretation of $\mathfrak{B}\varphi$ is ``holding belief $\varphi$''. So, a belief set $X$, which represents the whole set of beliefs held by the agent of some belief state, \emph{satisfies} $\mathfrak{B}\varphi$ if and only if $\varphi \in X$. Conditions of satisfaction for truth-functional combination of atomic descriptors are defined inductively in the usual way, and a belief set satisfies a descriptor if and only if it satisfies all its elements. 

Furthermore, for the construction of $\circ$, it assumes that there is an \textit{outcome set} $\mathbb{X}$ of  potential outcomes of belief changes, and each belief change is performed by a direct choice among these potential outcomes.  So, for example, choice revision $\cro$ can be constructed from descriptor revision $\circ$ in the way of $K \cro \{\varphi_0, \varphi_1, \cdots, \varphi_n\} = K \circ \{\mathfrak{B} \varphi_0 \vee \mathfrak{B} \varphi_1 \vee \cdots \vee \mathfrak{B} \varphi_n \}$, and the outcome of $K \circ \{\mathfrak{B} \varphi_0 \vee \mathfrak{B} \varphi_1 \vee \cdots \vee \mathfrak{B} \varphi_n \}$ is the set selected directly from the elements of $\mb{X}$ satisfying the descriptor $\{\mathfrak{B} \varphi_0 \vee \mathfrak{B} \varphi_1 \vee \cdots \vee \mathfrak{B} \varphi_n \}$. 

Given the intended interpretation of $\mathfrak{B}$, only belief sets could be taken as elements of the outcome set $\mathbb{X}$. If elements in $\mathbb{X}$ are sets of sentences which are not necessarily logically closed, i.e. \textit{belief bases},\footnote{ According to \cite{levi_subjunctives_1977,levi_fixation_1991}, a belief base $K$ could be interpreted as the set of statements the agent actually believes, and the corresponding belief set, i.e. the logical closure of $K$, should be interpreted as the set of statements that the agent is committed to believing. Also, the distinction between belief bases and belief sets has been related to the difference between ‘foundationalist’ and ‘coherentist’ viewpoints in epistemology \cite{gardenfors_dynamics_1990,sosa_raft_1980}. This article is not concerned with the issue of comparison between belief base and belief set. For more discussion on this topic, see \cite{hansson_textbook_1999}, p. 17-24.} then descriptors will become not appropriate for describing the success conditions of belief changes unless the interpretation of the metalinguistic symbol $\mathfrak{B}$ is correspondingly modified. However, this modification seems not straightforward.

In this contribution, in order to model choice revision $\cro$ (and the derived making up one's mind operation $\MU$) on belief bases, we choose another strategy rather than adapting the approach of descriptor revision to the context of belief bases. The methodology employed here is inspired by the work in \cite{hansson1993reversing}, where a \textit{package multiple revision}\footnote{ In package multiple revision, the new information is represented by a set of sentences and the agent always incorporates the entire new information into her beliefs.} is built up from two more elementary units: a \textit{package multiple contraction} and an \textit{expansion} operation. In this article, for the construction of choice revision on belief bases, in addition to the package multiple contraction, we will also make use of the \textit{choice multiple contraction} proposed in \cite{fuhrmann1994survey}. Moreover, we will propose a generalized expansion operation. With these building blocks in hand, we can start to construct two sorts of choice revision on belief bases.

The rest of this contribution will be structured as follows. In Section \ref{section preliminary}, we will present some formal preliminaries. In Section \ref{section tools for construction choice revision}, we will present the basics of package and choice multiple revision needed for the work in later sections. Moreover, in this section, we will introduce the generalized expansion operation mentioned above. In Section \ref{section constructing choice revision}, we will show how to construct two kinds of choice revision by using the tools previously introduced, and how to axiomatically characterize these constructions. Following this, Section \ref{section constructing making up one's mind} will be devoted to the study of the operations of making up one's mind derived from the two choice revision operations. Section \ref{section conclusions} concludes.

\section{Preliminaries}\label{section preliminary}

Let object language $\mathcal{L}$ be defined inductively from a set of propositional variables $\{p_0 ,\,  p_1, \, \cdots, \, p_n, \, \cdots \}$ and the truth-functional operations $\neg, \wedge, \vee$, $\rightarrow$ and $\leftrightarrow$. $\taut$ is a tautology and $\falsum$ a contradiction. Sentences in $\mathcal{L}$ will be denoted by lower-case Greek letters and sets of such sentences by upper-case Roman letters. 

$\cn$ is a consequence operation for $\mathcal{L}$ satisfying supraclassicality (if $\varphi$ can be derived from $A$ by classical truth-functional logic, then $\varphi \in \cons{A}$), compactness (if $\varphi \in \cons{A}$, then there exists some finite $B \subseteq A$ such that $\varphi \in \cons{B}$) and the deduction property ($\varphi \in \cn(A \cup \{\psi\})$ if and only if (henceforth ``iff'' for short) $\psi \rightarrow \varphi \in \cn(A)$). As we have mentioned, belief base is an arbitrary set $A$ of sentences, and a belief base $A$ is a belief set iff $A = \cons{A}$. 

For sets of sentences $A$ and $B$, $A \vdash B$ holds iff  $B \cap \cons{A} \neq \emptyset$,  and $A \Vdash B$ holds iff $B \subseteq \cons{A}$. We will omit the bracket of the set if it is singleton. For example, we will write $\varphi \vdash \psi$ instead of $\{\varphi\} \vdash \{\psi\}$. 

$A \equiv B$ holds iff for every $\varphi \in A$, there exists $\psi \in B$ such that $\emptyset \vdash \varphi \leftrightarrow \psi$ and vice versa. $A$ is a quotient-finite set iff there is a finite set $B$ such that $A \equiv B$.

$K \upVdash A$ and $K \angle A$ respectively denote the \textit{package remainder set } and \textit{choice remainder set} of $K$ with respect to $A$, which are formally defined as follows:

\begin{DEF}[\cite{hansson1993reversing}]
$X \in K \upVdash A$ iff (i) $X \subseteq K$, (ii) $X \nvdash A$, and (iii) For all $Y \subseteq K$, if $X \subset Y$, then $Y \vdash A$.
\end{DEF}

\begin{DEF}[\cite{fuhrmann1994survey}]
$X \in K \angle A$ iff (i) $X \subseteq K$, (ii) $X \nVdash A$, and (iii) For all $Y \subseteq K$, if $X \subset Y$, then $Y \Vdash A$.
\end{DEF}

A selection function $\gamma$ is a binary function defined as follows:

\begin{DEF}[\cite{alchourron_logic_1985,hansson1993reversing,fuhrmann1994survey}]$\,$
\begin{enumerate}
\item A selection function $\gamma$ is any function 
$\gamma : \mathcal{P}(\mathcal{L}) \times \mathcal{P}(\mathcal{P}(\mathcal{L}))  \to \mathcal{P}(\mathcal{P}(\mathcal{L}))$
such that $\emptyset \neq \gamma(X, \mathbb{Y}) \subseteq \mathbb{Y}$ for all $\mathbb{Y} \neq \emptyset$, and $\gamma(X, \mathbb{Y}) = \{ X \}$ otherwise.
\item A selection function $\gamma$ is $\bigtriangleup$-unified if and only if for all subsets $K_1$, $K_2$, $A_1$ and $A_2$  of $\mathcal{L}$: If $K_1 \bigtriangleup A_1 = K_2 \bigtriangleup A_2  \neq \emptyset$, then $\bigcap \gamma(K_1, K_1 \bigtriangleup A_1 ) = \bigcap \gamma(K_2, K_2 \bigtriangleup A_2)$, where $\bigtriangleup \in \{\angle ,\upVdash\}$.
\end{enumerate}
\end{DEF}

\section{Tools for constructing choice revision}\label{section tools for construction choice revision}

In this section, we exhibit the main components of which the choice revision on belief bases constructed here consists. As we have mentioned, they include two kinds of multiple contraction operations: \textit{package multiple contraction} \cite{hansson1993reversing} and \textit{choice multiple contraction} \cite{fuhrmann1994survey},\footnote{ Henceforth call them ``package contraction'' and ``choice contraction'' for short.} and a generalized version of expansion operation. 

\subsection{Package contraction and choice contraction}

Like choice revision, both package contraction and choice contraction are multiple belief changes. It means that in the process of these two belief changes, the questionable information  is represented by a set of sentences. The difference is that in package contraction the agent should remove all the questionable information, but in choice contraction the agent could remove its highly unbelievable part and at same time keep the rest relatively plausible part. In what follows, we perform a quick review of the formal properties of these two multiple contraction operations.

Package contraction operator is formally defined in the following way:
\begin{DEF}[\cite{hansson1993reversing}]
An operator $\pco : \mathcal{P}(\mathcal{L}) \times \mathcal{P}(\mathcal{L}) \to \mathcal{P}(\mathcal{L}) $ is a package contraction iff there exists a selection function $\gamma$ such that for every sets $K$ and $A$,
\[K \pco A = \bigcap\gamma(K, K \upVdash A).  \]
Moreover, $\pco$ is unified iff $\gamma$ is $\upVdash$-unified. 
\end{DEF}

It has been shown that package contraction (and unified package contraction) can be axiomatically characterized with a set of plausible postulates.

\begin{THE}[\cite{hansson1993reversing}]\label{representation theorem for package multiple contraction}$\,$
\begin{enumerate}
\item $\pco$ is a package contraction iff  it satisfies the following conditions: for all sets $K$ and $A$ ,

\begin{enumerate}
\item[] $K \pco A \subseteq K$ ($\pco$-inclusion).
\item[] If $\emptyset \nvdash A$, then $K \pco A \nvdash A$ ($\pco$-success).
\item[] If it holds for all $K^{\prime} \subseteq K$ that $K^{\prime} \vdash A$ iff $K^{\prime} \vdash B$, then $K \pco A = K \pco B$ ($\pco$-uniformity).  
\item[] If $\varphi \in K \setminus K \pco A$, then there is some $K^{\prime}$ with $K \pco A \subseteq K^{\prime} \subseteq K$, such that $K^{\prime} \nvdash A$ and $K^{\prime} \cup \{\varphi\} \vdash A$ ($\pco$-relevance).
\end{enumerate}

\item $\pco$ is unified iff it satisfies in addition the following:

\begin{enumerate}
\item[] If $\emptyset \nvdash A$, and each element of $Z$ implies an element of $A$, then $K \pco A = (K \cup Z) \pco A$ ($\pco$-redundancy).
\end{enumerate}  

\end{enumerate}

\end{THE}

Analogously, choice contraction is defined in terms of choice remainder sets and selection functions as follows:

\begin{DEF}[{\cite[modified]{fuhrmann1994survey}}]
An operator $\cco : \mathcal{P}(\mathcal{L}) \times \mathcal{P}(\mathcal{L}) \to \mathcal{P}(\mathcal{L}) $ is a choice contraction iff there exists a selection function $\gamma$ such that for every sets $K$ and $A$,
\[K \cco A = \bigcap\gamma(K, K \angle A).  \]
$\cco$ is unified iff $\gamma$ is $\angle$-unified. 
\end{DEF}


There is a partial representation theorem for the operation of choice contraction on belief bases.

\begin{THE}[{\cite[modified]{fuhrmann1994survey}}]\label{representation theorem for choice multiple contraction on belief sets}$\,$
\begin{enumerate}
\item $\cco$ is a choice contraction on belief bases with finite sets as inputs\footnote{\, In what follows, we will denote contraction/revision operator with a restricted domain, say $\mb{K} \times \mb{X}$, where both $\mb{K}$ and $\mb{X}$ are subsets of $\mathcal{P}(\mathcal{L})$, by \textit{contraction/revision on such and such sets (i.e. $\mb{K}$) with such and such inputs (i.e. $\mb{A}$)}.} iff it satisfies the following conditions: for all set $K$ and finite set $A$,
\begin{enumerate}
\item[] $K \cco A \subseteq A$ ($\cco$-inclusion).
\item[] If $ \emptyset \nVdash A$, then $K \cco A \nVdash A$ ($\cco$-success).
\item[] If it holds for all $K^{\prime} \subseteq K$ that $K^{\prime} \Vdash A$ iff $K^{\prime} \Vdash B$, then $K \cco A = K \cco B$ ($\cro$-uniformity).
\item[] If $\varphi \in K \setminus K \cco A$, then there is some $K^{\prime}$ with $K \cco A \subseteq K^{\prime} \subseteq K$, such that $K^{\prime} \nVdash A$ and $K^{\prime} \cup \{\varphi\} \vDash A$ ($\cco$-relevance).
\end{enumerate}

\item $\cco$ is unified iff it satisfies in addition the following:

\begin{enumerate}
\item[] If $\emptyset \nVdash A$, and each element of $Z$ implies all elements of $A$, then $K \cco A = (K \cup Z) \cco A$ ($\cco$-redundancy).
\end{enumerate}  

\end{enumerate}

\end{THE}

Theorem \ref{representation theorem for choice multiple contraction on belief sets} does not generally hold if we remove the restriction ``with finite sets as inputs''. For a counterexample, assume there is a set $A$ of infinitely atomic sentences $p_1, p_2, p_3, \cdots$ in the language $\mathcal{L}$ and let $K =  \{ p_1, p_1 \wedge p_2, p_1 \wedge p_2 \wedge p_3, \cdots \}$. Let $\cco$ be any  choice contraction on $K$.  We can see that $\cco$-success is violated: It holds that $\emptyset \nVdash A$, however, $K \angle A = \emptyset$ and hence $K = K \cco A   \Vdash A$. Nevertheless, \cite{fuhrmann1994survey} have shown that the set of postulates in Theorem \ref{representation theorem for choice multiple contraction on belief sets} axiomatically characterizes choice contractions on belief sets with \textit{arbitrary} inputs. To see why, we should first note that the following \textit{upper bound property} on package remainder set is indispensable for the proof of Theorem \ref{representation theorem for package multiple contraction}.

\begin{OBS}[{\cite[modified]{alchourron_hierarchies_1981}}]\label{observation on upper bound property}
Let $K$ be a belief base. For every $K^{\prime} \subseteq K$, if $K^{\prime} \nvdash A $, then there exists some set $X$ such that $K^{\prime} \subseteq X$ and $X \in K \upVdash A$.
\end{OBS}

Similarly, the following result on choice remainder set is needed for the proof of Theorem \ref{representation theorem for choice multiple contraction on belief sets}. 
 
\begin{OBS}\label{observation on partial upper bound property of choice remainder set}
Let $K$ be a belief base and $A$ a quotient-finite set. For every $K^{\prime} \subseteq K$, if $K^{\prime} \nVdash A $, then there exists some set $X$ such that $K^{\prime} \subseteq X$ and $X \in K \angle A$.
 
\end{OBS}

\begin{proof}[Proof for Observation \ref{observation on partial upper bound property of choice remainder set}:]
Let $A$ be a quotient-finite set, then there is a finite set $B$ such that $A \equiv B$. Moreover, let $K^{\prime} \subseteq K$ and let $K^{\prime} \nVdash A$. It follows from  $K^{\prime} \nVdash A$ that $A \neq \emptyset$. So $B$ is not empty neither. Hence, for every $X \subseteq \mathcal{L}$, $X \nVdash A $ holds iff $ X\nVdash B $ holds iff $ X \nvdash \& B $ holds, where $\& B$ is the conjunction of all propositions contained in $B$. It follows that $K^{\prime} \nvdash \& B$ and $K \upVdash \{\& B\} = K \angle A $. Moreover, by Observation \ref{observation on upper bound property}, it follows from $K^{\prime} \nvdash \& B$ and $K^{\prime} \subseteq K$ that there exists some set $X$ such that $K^{\prime} \subseteq X$ and $X \in K \upVdash \{\& B\}$. Thus, there exists some set $X$ such that $K^{\prime} \subseteq X$ and $X \in K \angle A$. 
\end{proof}

A choice remainder set $K \angle A$ does not generally have the upper bound property if we drop the condition that $A$ is quotient-finite. To see this, just let $K$ and $A$ be same as those in the previous example and let $K^{\prime} = \{ p_1\} $, then $K^{\prime} \nVdash A$ but since $K \angle A = \emptyset$ there is no $X \in K \angle A$ such that $K^{\prime} \subseteq X$. This fact explains the reason why Theorem \ref{representation theorem for choice multiple contraction on belief sets} cannot hold in general. By contrast, given that $K$ is logically closed, it holds for every set $A$ that $K \angle A$ has the upper bound property,\footnote{ \cite{fuhrmann1994survey} did not explicitly prove this result. It is not difficult to see that it follows from Observation \ref{observation on upper bound property} and an observation given in \cite{alchourron_logic_1985}: If $K$ is a belief set and $X \in K \upVdash \{\varphi\}$, then $X \in K \upVdash \{\psi\}$ for every $\psi \in K \setminus X$. The detailed proof is left to the reader.} which is necessary for  proving that the postulates in Theorem \ref{representation theorem for choice multiple contraction on belief sets} can exactly characterize all the choice contractions on belief sets.

\subsection{Partial expansion}

Comparing with revision and contraction, the construction of expansion operation $+$ is of a simpler form:
\[K + A = \cons{K \cup A}\] for expansion on belief sets, or 
\[K + A = K \cup A\] for expansion on belief bases.

Here we generalize the expansion operation to a non-prioritized version named \textit{partial expansion}, in which the agent could be expanding her set of beliefs with a part of the new information as well as dismissing the remaining. Before introducing the formal definition of partial expansion, we first define the notion of \textit{partial sum set} $K \Join A$ as follows:

\begin{DEF}
$X\in   K \Join A$ iff (i) $X \subseteq (K \cup A)$, (ii) $K \subseteq X$, and (iii) $X \cap A \neq \emptyset$.
\end{DEF}

We  observe that a partial sum set $K \Join A$ is not identical to any package or choice remainder set, except in the limiting case of that $K \Join A$ is a singleton.

\begin{OBS}\label{only need one selection function}
If $K_1 \Join A \neq \emptyset $ and there exist $K_2$ and $B$ such that $K_1 \Join A = K_2 \upVdash B \ $ or $K_1 \Join A = K_2 \angle B$, then $K_1 \Join A$ is a singleton.  
\end{OBS}

\begin{proof}[Proof for Observation \ref{only need one selection function}:]
Suppose towards contradiction that $K_1 \Join A$ is not a singleton. It follows from $K_1 \Join A \neq \emptyset$ that $A$ contains two distinct sentences $\varphi$ and $\psi$, of which at least one is not in $K$. By the definition of $\Join$, both $K_1 \cup\{\varphi\}$ and $K_1 \cup \{\varphi, \psi\}$ belong to $K_1 \Join A$. However, it is easy to see that there is no any package or choice remainder set which could contain both $K_1 \cup\{\varphi\}$ and $K_1 \cup \{\varphi, \psi\}$. Thus, $K_1 \Join A$ is a singleton.
\end{proof}

Next we define the following selection functions with respect to partial sum sets:

\begin{DEF} 
 Let $\gamma$ be a selection function. 
\begin{enumerate}
\item $\gamma$ is $\Join$-choice-full iff $\gamma(K, K \Join A) = K \Join A$ for all $K$ and $A$ such that $K \Join A \neq \emptyset$.
\item $\gamma$ is $\Join$-unified iff for all $K_1$, $K_2$, $A_1$ and $A_2$: If $K_1 \Join A_1 = K_2 \Join A_2  \neq \emptyset$, then $\bigcup \gamma(K_1, K_1 \Join A_1 ) = \bigcup \gamma(K_2, K_2 \Join A_2)$.
\item $\gamma$ is $\Join$-consistency-preserved iff for all $K$ and $A$: If there exists a set $X \in K \Join A$ such that $X \nvdash \falsum$, then $\bigcup \gamma(K, K \Join A) \nvdash \falsum$.
\end{enumerate}
\end{DEF}

Now a range of partial expansion operations can be defined in terms of such selection functions.

\begin{DEF}[Partial expansion]
An operator $\meo : \mathcal{P}(\mathcal{L}) \times \mathcal{P}(\mathcal{L}) \to \mathcal{P}(\mathcal{L}) $ is a partial expansion iff there exists a selection function $\gamma$ such that for every sets $K$ and $A$,
\[K \meo A = \bigcup\gamma(K, K \Join A). \]
$\meo$ is choice-full, unified or consistency-preserved iff $\gamma$ is $\Join$-choice-full, $\Join$-unified or $\Join$-consistency-preserved respectively. 
\end{DEF}

\noindent Partial expansion on belief sets can be defined in almost the same way:
\[K \meo A = \cons{\bigcup\gamma(K, K \Join A)}.\]
It is easy to see that the normal expansion operation is a special case of partial expansion, namely the choice-full partial expansion. Moreover, it follows immediately from the following observation that all partial expansions are unified.

\begin{OBS}\label{Observation that all selection functions are Join-unified}
All selection functions are $\Join$-unified. 
\end{OBS}

\begin{proof}[Proof for Observation \ref{Observation that all selection functions are Join-unified}:]
It is enough to show that if $K_1 \Join A = K_2 \Join B \neq \emptyset$ and $K_1 \neq K_2$, then both $K_1 \Join A$ and $K_2 \Join B$ are singletons. Suppose $A \cap K_1 \neq \emptyset$ and $B \cap K_2 \neq \emptyset$. Then, $K_1 \in K_1 \Join A$. So, by our assumption that $K_1 \Join A = K_2 \Join B \neq \emptyset$, $K_1 \in K_2 \Join B$ whence $K_2 \subseteq K_1$. So $K_2 \subset K_1$ since $K_1 \neq K_2$. Similarly, it can be shown that $K_2 \subset K_1$, which leads to a contradiction. So it holds that either $A \cap K_1 = \emptyset$ or $B \cap K_2 = \emptyset$. Without loss of generality, let us assume that $A \cap K_1 = \emptyset$. Suppose $K_1 \Join A \neq \emptyset$ is not a singleton, then there exist two distinct sentences $\varphi$ and $\psi$ in $A$. It follows that $K_1 \cup \{\varphi\} \in K_1 \Join A$ and $K_1 \cup \{\psi\} \in K_1 \Join A$. So $K_1 \cup \{\varphi\} \in K_2 \Join B$ and $K_1 \cup \{\psi\} \in K_1 \Join B$ since $K_1 \Join A = K_2 \Join B$. So $K_2 \subseteq K_1 \cup \{\varphi\}$ and $K_2 \subseteq K_1 \cup \{\psi\}$, whence $K_2 \subseteq K_1$. But $K_1 \neq K_2$, then $K_2 \subset K_1$. Moreover, it follows from $K_1 \Join A = K_2 \Join B$ that $\varphi \in K_2 \cup B$. So $\varphi \in B$, since $K_2 \subset K_1$ and $\varphi \in A$ which is disjoint with $K_1$. It follows that $K_2 \cup \{\varphi\} \in K_2 \Join B$, whence $K_2 \cup \{\varphi\} \in K_1 \Join A$. However, it is easy to see that $K_1 \not \subseteq K_2 \cup \{\varphi\}$, which leads to a contradiction. Thus, both $K_1 \Join A$ and $K_2 \Join B$ are singletons.
\end{proof}

An axiomatic characterization of partial expansion is obtainable with a set of postulates of the following simple forms:

\begin{THE}[Representation theorem for partial expansion]\label{Representation theorem for general expansion}
$\,$
\begin{enumerate}
\item $\meo$ is a partial expansion iff it satisfies the following postulates:

\begin{enumerate}
\item[] $K \meo A \subseteq K \cup A$ (\textbf{$\meo$-inclusion}).
\item[] $K \subseteq K \meo A$ (\textbf{$\meo$-preservation}).
\item[] If $A \neq \emptyset$, then $A \cap (K \meo A) \neq \emptyset $ (\textbf{$\meo$-success}).
\item[] If $K \cap A \neq \emptyset$ and $A \subseteq B \subseteq (K \cup A)$, then $K \meo A = K \meo B$ (\textbf{$\meo$-coincidence}). 
\end{enumerate}

\item $\meo$ is consistency-preserved  iff it in addition satisfies:

\begin{enumerate}
\item[] If there exists $X$ such that $K \subseteq X \subseteq (K \cup A)$ and $X \nvdash \falsum$, then $K \meo A \nvdash \falsum$ (\textbf{$\meo$-consistency}).
\end{enumerate}

\end{enumerate}

\end{THE}

See appendix for the proofs of this and other representation theorems.

\section{Constructing choice revision}\label{section constructing choice revision}

Now we turn to the construction of choice revision $\cro$ on belief bases. Besides the operations discussed in the previous section, the following  construction of \textit{negation set}  will also be used.   

\begin{DEF}[{\cite[modified\footnote{\, In the original definition of negation set, $n(\emptyset)$ is defined as $\falsum$ instead of $\taut$. We modify it since it is intuitive to let $K \cro A = K$ in the limiting case of that $A = \emptyset$. This would be easily realized if we use the modified definition of negation set.}]{hansson1993reversing}}]
Let $A$ be some set of sentences. Then the negation set $n(A)$ of $A$ is defined as follows:
\begin{enumerate}
\item $n(\emptyset) = \taut$,
\item $n(A) = \underset{n \geq 1}{\bigcup} \{\neg \varphi_1 \vee \neg \varphi_2 \vee \cdots \vee \neg \varphi_n \mid \varphi_i \in A \mbox{ for every } i \mbox{ such that } 1 \leq i \leq n \}$. 
\end{enumerate}
\end{DEF}

Negation set extends the notion of negation of sentence. It is obvious that $n(A)$ is quotient-finite when $A$ is finite.

In what follows, two kinds of choice revision operators on belief bases will be constructed, namely \textit{internal choice revision} and \textit{external choice revision}. These two terms are based on the terminology used in \cite{hansson1993reversing} for package multiple revision.

\subsection{Internal choice revision}

\begin{DEF}[Internal choice revision]\label{definition of internal choice revision}
An operator $\cro : \mathcal{P}(\mathcal{L}) \times \mathcal{P}(\mathcal{L}) \to \mathcal{P}(\mathcal{L}) $ is an internal choice revision iff there exists a choice contraction $\cco$ and a consistency-preserving partial expansion $\meo$ such that for every sets $K$ and $A$,
\[K \cro A = K \cco n(A) \meo A. \]
$\cro$ is unified iff $\meo$ and $\cco$ are unified.
\end{DEF}

The intuitive meaning of Definition \ref{definition of internal choice revision} is that the process of belief change of choice revision could be decomposed into two steps: In order to revise $K$ to incorporate some subset of the new information $A$, the agent could first contract $K$ by a set of sentences which contradicts some part of $A$, then expand the contracted belief base with a subset of $A$ which is consistent with it. 

Although the procedure of internal choice revision is divided into two seemly independent steps, we observe that every internal choice revision can be reconstructed by a single selection function.

\begin{OBS}\label{observation that internal choice revision can be reconstructed by a single selection function}
$\cro$ is an internal choice revision iff there exists a $\Join$-consistency preserving selection function $\gamma$ such that:
\[K \cro A = \bigcup\gamma(\bigcap \gamma(K, K \angle n(A)), (\bigcap \gamma(K, K \angle n(A))) \Join A).\]
Moreover, $\cro$ is unified iff $\gamma$ is $\angle$-unified.
\end{OBS}

\begin{proof}[Proof for Observation \ref{observation that internal choice revision can be reconstructed by a single selection function}:]
It follows immediately from the definitions of operations used in the construction of internal choice revision, and Observations \ref{only need one selection function} and \ref{Observation that all selection functions are Join-unified}.
\end{proof}

The following observation tells us that if the original belief base $K$ is consistent, then the overlapping parts of $K$ and the new information $A$ will be kept in the outcome of choice contraction $\cco$ on $K$ by the negation set $n(A)$ of $A$.

\begin{OBS}\label{lemma for internal choice revision}
Let $\cco$ be a choice contraction. If $K$ is consistent, then $K \cap A \subseteq K \cco n(A)$ for every $A$.
\end{OBS}

\begin{proof}[Proof for Observation \ref{lemma for internal choice revision}:]
Suppose towards contradiction that there exists $\varphi$ such that $\varphi \in K \cap A$ and $\varphi \notin K \cco n(A)$. By $\varphi \in K$ and the definition of $\cco$, it follows from $\varphi \notin K \cco n(A)$ that there exists $X \in K \angle n(A)$ such that $\varphi \notin X$. Since $\varphi \in K \cap A$, by the definition of choice remainder, $X \cup \{\varphi\} \vdash \neg \varphi$ and hence  $X \vdash \neg \varphi$. So $K \vdash \neg \varphi$. But it contradicts that $\varphi \in K$ and $K$ is consistent. Thus, $K \cap A \subseteq K \cco n(A)$.
\end{proof}

So, suppose that $\cro$ is an internal choice revision based on some choice contraction $\cco$ and some partial expansion, it is not difficult to observe that $\cro$ has the following property:

\begin{OBS}\label{observation for internal choice revision}
Let $\cco$ be a choice contraction and $\cro$ an internal choice revision based on $\cco$. If $K$ is consistent, then $K \cap (K \cro A) = K \cco n(A)$ for every $A$.
\end{OBS}

\begin{proof}[Proof for Observation \ref{observation for internal choice revision}:]
By $\cco$-inclusion and $\meo$-inclusion, $K \cco n(A) \subseteq K$ and $K \cco n(A) \subseteq K \cco n(A) \meo A = K \cro A$, so $K \cco n(A) \subseteq K \cap (K \cro A)$. For the other inclusion direction, $K \cap (K \cro A) = K \cap (K \cco n(A) \meo A) \subseteq K \cap (K \cco n(A) \cup A)$, whence $K \cap (K \cro A) \subseteq (K \cap (K \cco n(A)) \cup (K \cap A) =( K \cco n(A)) \cup (K \cap A)$. By Observation \ref{lemma for internal choice revision}, $K \cap A \subseteq K \cco n(A)$ when $K$ is consistent. It follows that $K \cap (K \cro A) \subseteq  K \cco n(A)$. Thus, $K \cap (K \cro A) \subseteq  K \cco n(A)$ when $K$ is consistent.  
\end{proof}

Observation \ref{observation for internal choice revision} together with Theorems \ref{representation theorem for choice multiple contraction on belief sets} and \ref{Representation theorem for general expansion} hint that we can provide an axiomatic characterization for a restricted variant of internal choice revision with the following postulates: 

\begin{THE}[Partial representation theorem for internal choice revision]\label{representation theorem for internal choice revision}$\,$
\begin{enumerate}
\item  $\cro$ is an internal choice revision on consistent belief bases with finite inputs iff it satisfies the following conditions: for every consistent $K$ and finite $A $ and $B$,
\begin{enumerate}
\item[] $K \cro A \subseteq (K \cup A)$ (\textbf{$\cro$-inclusion}).
\item[] If $A \neq \emptyset$, then $A \cap (K \cro A) \neq \emptyset $ (\textbf{$\cro$-success}).
\item[] $K \cro A = (K \cap (K \cro A)) \cro A$ (\textbf{$\cro$-iteration}).
\item[] If $A \not \equiv \{ \perp \}$, then $K \cro A \nvdash \falsum$  (\textbf{$\cro$-consistency}).
\item[] If $A \cap K \neq \emptyset$ and $A \subseteq B \subseteq (A \cup K) $, then $K \cro A = K \cro B$ (\textbf{$\cro$-coincidence}).
\item[] If it holds for all $K^{\prime} \subseteq K$ that $K^{\prime} \cup \{\varphi\} \nvdash \falsum$ for some $\varphi \in A$ iff $K^{\prime} \cup \{\psi\} \nvdash \falsum$ for some $\psi \in B$, then $K \cap (K \cro A) = K \cap (K \cro B$) (\textbf{$\cro$-uniformity}).
\item[] If $\varphi \in K \setminus K \cro A$, then there is some $K^{\prime}$ with $K \cap (K \cro A) \subseteq K^{\prime} \subseteq K $, such that $K^{\prime} \cup \{\psi\} \nvdash \falsum$ for some $\psi \in A$ and $K^{\prime} \cup \{\varphi\} \cup \{\lambda\} \vdash \falsum $ for every $\lambda \in A$ (\textbf{$\cro$-relevance}).
\end{enumerate}

\item $\cro$ is additionally unified iff it satisfies in addition the following:

\begin{enumerate}
\item[] If $K \cup Z \nvdash \falsum$, $A \not \equiv \{ \perp \}$, $A \neq \emptyset$ and it holds for every $\varphi \in Z$ that $\varphi \vdash \neg \psi$ for all $\psi \in A$ , then $K \cro A = (K \cup Z)\cro A$ (\textbf{$\cro$-redundancy}).
\end{enumerate}  

\end{enumerate}
\end{THE}

With the exception of $\cro$-iteration, each of these postulates corresponds to one of the postulates that characterize choice contraction or partial expansion. According to $\cro$-iteration,  let $K^{\prime}$ be the remaining part of $K$ after performing an internal choice revision on $K$, then performing again the choice revision on $K^{\prime}$ will generate the same outcome as performing the choice revision on $K$. This is in conformity with the ``first contract, then expand'' process of the belief change of internal choice revision.

\subsection{External choice revision}

In this subsection, we investigate an alternative construction of choice revision on belief bases. Generally speaking, this construction originates from the idea that we can reverse the procedure of internal choice revision: In order to revise $K$ to take some sentences in $A$, the agent could first properly expand the $K$ with a subset $A^{\prime}$ of $A$, then contract the negation set $n(A^{\prime})$ of $A^{\prime}$. This methodology is not workable in the context of belief sets, since it typically involves a temporary inconsistent belief state, and different inconsistent belief states cannot be distinguished by the only one inconsistent belief set, i.e. $\conp{\falsum}$. The choice revision constructed in this approach is called \textit{external choice revision}. Formally, it is defined as follow:

\begin{DEF}[External choice revision]
An operator $\cro : \mathcal{P}(\mathcal{L}) \times \mathcal{P}(\mathcal{L}) \to \mathcal{P}(\mathcal{L}) $ is an external choice revision iff there exists a package contraction $\pco$ and a partial expansion $\meo$ such that for all $K$ and $A$,

\[K \cro A = K \meo A \pco n(A^{\prime})\]
where $A^{\prime} = (K \meo A) \setminus K$. Moreover, $\cro$ is unified iff $\meo$ and $\pco$ are unified. 
\end{DEF}

Similar to choice revision, we observe that every external choice revision can be reconstructed from a single selection function.

\begin{OBS}\label{observation that external choice revision can be reconstructed by a single selection function}
$\cro$ is an external choice revision iff there exists a selection function $\gamma$ such that:
\[K \cro A = \bigcap\gamma(\bigcup \gamma(K, K \Join A), (\bigcup \gamma(K, K \Join A)) \upVdash n(A^{\prime})),\]
where $A^{\prime} = (\bigcup \gamma(K, K \Join A)) \setminus K $. Moreover, $\cro$ is unified iff $\gamma$ is $\upVdash$-unified.
\end{OBS}

\begin{proof}[Proof for Observation \ref{observation that external choice revision can be reconstructed by a single selection function}:]
It follows immediately from the definitions of operations used in the construction of external choice revision, and Observations \ref{only need one selection function} and \ref{Observation that all selection functions are Join-unified}.
\end{proof}

For axiomatization of external choice revision, we first prove the following observation, which plays a role for external choice revision similar to that  Observation \ref{observation for internal choice revision} plays for internal choice revision. 

\begin{OBS}\label{lemma for external choice revision}
Let $\meo$ be a partial expansion, $\pco$ a package  contraction and $\cro$ the external choice revision generated from them. Then,
\begin{enumerate}
\item $(K \meo A) \setminus K = (K \cro A) \setminus K$;
\item $(K \cro A) \cup K = K \meo A$.
\end{enumerate}
\end{OBS}

\begin{proof}[Proof for Lemma \ref{lemma for external choice revision}:]
\textit{1.} Let $\Ap = (K \meo A) \setminus K$. Since $K \cro A = (K \meo A) \pco n(\Ap)$, $K \cro A \subseteq K \meo A$ by $\pco$-inclusion. So $(K \cro A) \setminus K \subseteq \Ap$ by basic set theory. Since $\Ap \cap K = \emptyset$, in order to complete the proof, we only need to prove that $\Ap \subseteq K \cro A $. If $A^{\prime} = \emptyset$, it holds trivially. If $A^{\prime} \neq \emptyset$, we should consider two cases. (i) $\Ap$ is inconsistent. In this case, $(K \meo A)\upVdash n(\Ap) = \emptyset$. So $K \cro A = (K \meo A) \pco n(\Ap)=K \meo A $ and hence $\Ap \subseteq K \cro A$. (ii) $\Ap$ is consistent. In this case, $(K \meo A)\upVdash n(\Ap) \neq \emptyset$. We should show that $A^{\prime} \subseteq X$ for every $X \in (K \meo A) \upVdash n(A^{\prime}) $. Suppose towards contradiction that there exists $X \in (K \meo A) \upVdash n(A^{\prime}) $ such that $\Ap \not \subseteq X$, i.e. there exists $\varphi$ such that $\varphi \in \Ap$ and $\varphi \notin X$. Then $X, \varphi \vdash n(\Ap) $ by the definition of package remainder set. It follows from the definition of negation set that there exist $\beta_1 , \cdots, \beta_n \in A$ such that $X, \varphi \vdash \neg \beta_1 \vee \cdots \beta_n$. So, by the deduction property, $X \vdash \varphi \to (\neg \beta_1 \vee \cdots \beta_n)$, i.e. $X \vdash \neg \varphi \vee \neg \beta_1 \vee \cdots \beta_n$. It follows that $X \vdash n(\Ap)$ since $\varphi \in A$. However, it contradicts that $X \nvdash n(\Ap)$ as $X \in (K \meo A) \upVdash n(A^{\prime}) $. Thus, $\Ap = (K \meo A) \setminus K = (K \cro A) \setminus K$.\\
\textit{2.}  It follows from $(K \meo A) \setminus K = (K \cro A) \setminus K$ that $(K \cro A) \cup K  = ((K \cro A) \setminus K) \cup K = ((K \meo A) \setminus K) \cup K $. Since $K \subseteq K \meo A$, $((K \meo A) \setminus K) \cup K = K \meo A$. Thus, $(K \cro A) \cup K = K \meo A$.
\end{proof}

So, as suggested by the representation theorems on package contraction and partial expansion, we can show that external choice revision can be axiomatically characterize by the following postulates.

\begin{THE}[Representation theorem for external choice revision]\label{representation theorem for external choice revision}
$\,$
\begin{enumerate}
\item  $\cro$ is an external choice revision iff it satisfies the following conditions: for all $K$, $K_1$, $K_2$, $A$ and $B$,
\begin{enumerate}
\item[] $K \cro A \subseteq (K \cup A)$ (\textbf{$\cro$-inclusion}).
\item[] If $A \neq \emptyset$, then $A \cap (K \cro A) \neq \emptyset$ (\textbf{$\cro$-success}).
\item[] If $(A \cap (K \cro A)) \subseteq K$, then $K \cro A = K$ (\textbf{$\cro$-confirmation}).
\item[] If $(K \cro A) \setminus  K \neq \emptyset$ and $(K \cro A) \setminus  K \nvdash \falsum$, then $K \cro A \nvdash \falsum$ (\textbf{$\cro$-Consistency}).
\item[] If $A \cap K \neq \emptyset$ and $A \subseteq B \subseteq (A \cup K) $, then $K \cro A = K \cro B$ (\textbf{$\cro$-coincidence}).
\item[] If $K_1 \neq ((K_1 \cro A) \cup K_1 ) =K = ((K_2 \cro B) \cup K_2) \neq K_2$ and it holds for all $K^{\prime} \subseteq K$ that $K^{\prime} \cup ((K_1 \cro A) \setminus K_1) \nvdash \falsum$   iff $K^{\prime} \cup  ((K_2 \cro B) \setminus K_2) \nvdash \falsum$, then $K_1 \cro A = K_2 \cro B$ (\textbf{$\cro$-Uniformity}).
\item[] If $\varphi \in K \setminus K \cro A$, then there is some $K^{\prime}$ with $K \cro A \subseteq K^{\prime} \subseteq K \cup (K \cro A)$, such that $K^{\prime} \nvdash \falsum$  and $K^{\prime} \cup \{\varphi\} \vdash \falsum$ (\textbf{$\cro$-Relevance}).
\end{enumerate}

\item $\cro$ is additionally unified iff it satisfies in addition the following:

\begin{enumerate}
\item[] If it holds for all $X \subseteq \mathcal{L}$ that $X \subseteq K_1 \cup (K_1 \cro A)$ and $X \cup ((K_1 \cro A) \setminus K_1)$ is consistent iff $X \subseteq K_2 \cup (K_2 \cro B)$ and $X \cup ((K_2 \cro B) \setminus K_2)$ is consistent, then $K_1 \cro A = K_2 \cro B$. (\textbf{$\cro$-strong Uniformity}).

\end{enumerate}  

\end{enumerate}

\end{THE}

Only $\cro$-confirmation has no correspondent in the set of  postulates which characterizes package contraction or partial expansion. The justification of $\cro$-confirmation is that if the part of the new information which the agent thinks plausible and intends to accept is already in the her original belief state, then nothing needs to be done to incorporate it. If $\cro$ is the sort of internal choice revision discussed in Theorem \ref{representation theorem for internal choice revision}, then $\cro$ also satisfies $\cro$-confirmation: $K \cro A \subseteq K$ follows immediately from $\cro$-inclusion, and $K \subseteq K \cro A $ follows immediately from $\cro$-relevance under the condition that $K$ is consistent.

The following two examples confirm that internal choice revision and external choice revision are indeed different operations: Let $p$, $q$ and $r$ be pairwise distinct atomic propositions,

\begin{enumerate}
\item[] \textbf{Example 1}: Let $K = \{p, \neg q, \neg r \}$ and $A = \{q, r \}$. So $K \angle n(A) =\{\{p, \neg q\}, \linebreak \{p, \neg r\}\}$. Let $\gamma$ be some selection function satisfying that \linebreak $\gamma(K, K \angle n(A)) = K \angle n(A)$ and $ \gamma (\{p\}, \{p\} \Join A) = \{\{p, q\}\}$ and let $\cro$ be the internal choice revision derived from $\gamma$. Then,  $K \cro A = \{p, q\} $. It is easy to check that $\cro$ does not satisfy $\cro$-Relevance. So, by Theorem \ref{representation theorem for external choice revision}, $\cro$ is not an external choice revision.
\\
\item[] \textbf{Example 2}: Let $K = \{q\}$ and $A = \{p, p \to \neg q\}$. Let $\gamma$ be some selection function satisfying that $\gamma (K , K \Join A) = K \Join A$ and let $\cro$ be the external choice revision derived from $\gamma$. Note that $(K \cup A) \upVdash n(A) = \{\{p, p\to \neg q\}\}$. It follows that $K \cro A = \{p, p\to \neg q\}$. It is easy to see that $\cro$ does not satisfy $\cro$-relevance. Moreover, $K$ is consistent and $A$ is finite. So, by Theorem \ref{representation theorem for internal choice revision}, $\cro$ is not an internal choice revision.
\end{enumerate}

\section{Constructing making up one's mind}\label{section constructing making up one's mind}

As we have mentioned, the operation of making up one's mind can be constructed from choice revision in a straightforward way. So, two sorts of making up one's mind operations are naturally generated from the internal and external choice revision as follows.

\begin{DEF}[Making up one's mind]$\,$
\begin{enumerate}
\item An operation $\MU : \mathcal{P}(\mathcal{L}) \times \mathcal{L} \to \mathcal{P}(\mathcal{L}) $ is a (unified) internal making up one's mind operation iff there exists a (unified) internal choice revision $\cro$ such that for every set $K$ and sentence $\varphi$,
\[K \MU \varphi = K \cro \{\varphi , \neg \varphi \}. \]
\item An operation $\MU : \mathcal{P}(\mathcal{L}) \times \mathcal{L} \to \mathcal{P}(\mathcal{L}) $ is a (unified) external external making up one's mind operation iff there exists a (unified) external choice revision $\cro$ such that for every set $K$ and sentence $\varphi$,
\[K \MU \varphi = K \cro \{\varphi , \neg \varphi \}. \]
\end{enumerate}
\end{DEF}

Also, we can obtain the following representation theorems for these two generated making up one's mind operations. Note that though we only provide a partial representation theorem for internal choice revision, a full representation theorem for internal making up one's mind operation is obtainable.

\begin{THE}[Representation theorem for internal making up one's mind]\label{representation theorem for internal making up one's mind}
$\,$
\begin{enumerate}
\item  $\MU$ is an internal making up one's mind operation  iff it satisfies the following conditions:
\begin{enumerate}
\item[] $K \MU \varphi \subseteq (K \cup \{\varphi, \, \neg \varphi\})$ (\textbf{$\MU$-inclusion}).
\item[] $\varphi \in K \MU \varphi $ or $ \neg \varphi \in K \MU \varphi$ (\textbf{$\MU$-success}).
\item[] $\falsum \notin K \MU \varphi$ (\textbf{$\cro$-consistency}).
\item[] If $K \nvdash \falsum$, $\{\varphi, \neg \varphi\} \cap K \neq \emptyset$, $\{\psi, \neg \psi\} \cap K \neq \emptyset$ and $K \cup \{\varphi, \neg \varphi\}  = K \cup \{\psi, \neg \psi\} $, then $K \MU \varphi = K \MU \psi$ (\textbf{$\MU$-coincidence}).
\item[] $K \MU \varphi = ((K \MU \taut) \cap K) \MU \varphi$ (\textbf{$\MU$-iteration}).
\item[] If $\psi \in K \setminus K \MU \varphi$, then there is some $K^{\prime}$ with $(K \MU \taut) \cap K \subseteq K^{\prime} \subseteq K $, such that $K^{\prime} \nvdash \falsum$ and $K^{\prime} \cup \{\psi\} \vdash \falsum$  (\textbf{$\MU$-relevance}).
\end{enumerate}

\item $\MU$ is additionally unified iff it satisfies in addition the following:

\begin{enumerate}
\item[] If $Z \equiv \{ \perp \}$, then $K \MU \varphi = (K \cup Z)\MU \varphi$ (\textbf{$\MU$-redundancy}).
\end{enumerate}  

\end{enumerate}

\end{THE}

\begin{THE}[Representation theorem for external making up one's mind]\label{representation theorem for external making up one's mind}$\,$
\begin{enumerate}
\item  $\MU$ is an external making up mind operation iff it satisfies the following conditions:
\begin{enumerate}
\item[] $K \MU \varphi \subseteq (K \cup \{\varphi, \, \neg \varphi\})$ (\textbf{$\MU$-inclusion}).
\item[] $\varphi \in K \MU \varphi $ or $ \neg \varphi \in K \MU \varphi$ (\textbf{$\MU$-success}).
\item[] If $(\phinegphi \cap (K \MU \varphi)) \subseteq K$, then $K \MU \varphi = K$ (\textbf{$\MU$-confirmation}).
\item[] If $(K \MU \varphi) \setminus  K \neq \emptyset$  and $(K \MU \varphi) \setminus  K \nvdash \falsum$, then $K \MU \varphi \nvdash \falsum$ (\textbf{$\MU$-Consistency}).
\item[] If $\{\varphi, \neg \varphi\} \cap K \neq \emptyset$, $\{\psi, \neg \psi\} \cap K \neq \emptyset$ and $K \cup \{\varphi, \neg \varphi\}  = K \cup \{\psi, \neg \psi\} $, then $K \MU \varphi = K \MU \psi$ (\textbf{$\MU$-Coincidence}).
\item[] If $K_1 \neq ((K_1 \MU \varphi) \cup K_1 ) =K = ((K_2 \MU \psi) \cup K_2) \neq K_2$ and it holds for all $K^{\prime} \subseteq K$ that $K^{\prime} \cup ((K_1 \MU \varphi) \setminus K_1) \vdash \falsum$   iff $K^{\prime} \cup  ((K_2 \MU \psi) \setminus K_2) \vdash \falsum$ , then $K_1 \MU \varphi = K_2 \MU \psi$ (\textbf{$\cro$-Uniformity}).
\item[] If $\psi \in K \setminus K \MU \varphi$, then there is some $K^{\prime}$ with $K \MU \varphi \subseteq K^{\prime} \subseteq K \cup (K \MU \varphi)$, such that $K^{\prime} \vdash \falsum$   and $K^{\prime} \cup \{\psi\} \nvdash \falsum$  (\textbf{$\MU$-Relevance}).
\end{enumerate}

\item $\MU$ is additionally unified iff it satisfies in addition the following:

\begin{enumerate}
\item[] If it holds for all $X \subseteq \mathcal{L}$ that $X \subseteq K_1 \cup (K_1 \MU \varphi)$ and $X \cup ((K_1 \MU \varphi) \setminus K_1) \vdash \falsum$  iff $X \subseteq K_2 \cup (K_2 \MU \psi)$ and $X \cup ((K_2 \MU \psi) \setminus K_2) \vdash \falsum$, then $K_1 \MU \varphi = K_2 \MU \psi$ (\textbf{$\MU$-strong Uniformity}).
\end{enumerate}  

\end{enumerate}
\end{THE}

The difference between the internal  and the external making up one's mind operations can be shown by using the following two examples: Let $p_1$, $p_2$, $p_3$ be pairwise distinct atomic propositions,

\begin{enumerate}
\item[] \textbf{Example 3:} Let $K^{\prime} = A = \{p_1, \neg p_1\}$ and let $\cro$ be any internal choice revision. Moreover, let $\MU$ be the internal making up one's mind derived from $\cro$. By the construction of internal choice revision, it is easy to see that either $K \MU p_1 = K \cro A = \{p_1\}$ or $K \MU p_1 = K \cro A = \{\neg p_1\}$. So $\MU$ cannot coincide with any external making up one's mind operation since it violates that $\MU$-confirmation.
\\
\item[] \textbf{Example 4:} Let $K = \{p_1, p_2\}$ and let $\gamma$ be some selection function satisfying that $\bigcup\gamma(K, K \Join \{p_3, \neg p_3\}) = K \cup \{p_3, \neg p_3\} $ and let $\cro$ be external choice revision derived from $\gamma$. Then, the external making up one's mind operation $\MU$ derived from $\cro$ is not an internal choice revision since $K \MU p_3 = K \cro \{p_3, \neg p_3\} = K \cup \{p_3, \neg p_3\}$ which violates $\MU$-consistency.

\end{enumerate}

The last example seemly also suggests that, as a modelling for the operation of making up one's mind,  the internal making up one's mind operators are better than the external operators since the latter violates $\MU$-consistency. Intuitively, to making up one's mind about a sentence $\varphi$ is usually understood as to accept \textit{either} $\varphi$ \textit{or} $\neg \varphi$. Accepting both $\varphi$ and $\neg \varphi$, just like neither believing $\varphi$ nor believing $\neg \varphi$, reflects that the agent has not made a decision between $\varphi$ and $\neg \varphi$.

\section{Conclusions}\label{section conclusions}

In this article we investigated the formal properties of choice revision, \linebreak a sort of non-prioritized multiple revision, in the context of belief bases. We presented constructions of two families of choice revision operators, namely internal choice revision and external choice revision. The constructions are based on two multiple contraction operations already known in the literature and a generalized expansion operation first introduced in this article. This approach to choice revision is different from the method employed in \cite{ 2018arXiv180407602Z}, where the choice revision is modelled by descriptor revision. We have indicated that, at least in its original form, descriptor revision is not suitable for modelling belief changes on belief bases.

Also, we investigated the axiomatic characterizations of the constructed choice revision. For internal choice revision, we provided a partial representation theorem (Theorem \ref{representation theorem for internal choice revision}) for the subset of this kind of operators on consistent belief bases with finite sets as inputs. For external choice revision, a full representation theorem (Theorem \ref{representation theorem for external choice revision}) was obtained. Furthermore, we studied two sorts of making up one's mind operators derived from these two choice revision operators and also proved the corresponding representation theorems (Theorems \ref{representation theorem for internal making up one's mind} and \ref{representation theorem for external making up one's mind}).

For future work, we want to extend Theorems \ref{representation theorem for choice multiple contraction on belief sets} and \ref{representation theorem for internal choice revision} to cover the general case. Furthermore, properties of choice revision operators generated from selection functions with more \textit{relational constraints}, such as those proposed by \cite{alchourron_logic_1985}, remain to investigate. Finally, since the method of construction of internal choice  revision is also applicable to operations on belief sets, it is interesting to study the difference and connection between the internal choice revision on belief sets and the choice revision modelled by descriptor revision.

\section*{Appendix: proofs}

We first give a lemma which is useful in the later proofs.

\begin{LEM}\label{lemma on interderivability between postulates}
Let $\meo$ be  a partial expansion operator and $\cro$ a choice revision operator.
\begin{enumerate}

\item $\meo$ satisfies $\meo$-coincidence iff  it satisfies:
\begin{enumerate}
\item[] If $A \cap K \neq \emptyset$, $B \cap K \neq \emptyset$ and $K \cup A = K \cup B$, then $K \meo A = K \meo B$ (\textbf{$\meo$-Coincidence}).
\end{enumerate} 

\item $\cro$ satisfies $\cro$-coincidence iff it satisfies:
\begin{enumerate}
\item[] If $A \cap K \neq \emptyset$, $B \cap K \neq \emptyset$ and $K \cup A = K \cup B$, then $K \cro A = K \cro B$ (\textbf{$\cro$-Coincidence}).
\end{enumerate} 

\item If $\cro$ satisfies $\cro$-inclusion and $\cro$-relevance, then it satisfies:
\begin{enumerate}
\item[] If $A = \emptyset$, then $K \cro A = K$ (\textbf{$\cro$-vacuity}). 
\end{enumerate}

\item If $\cro$ satisfies $\cro$-consistency and $\cro$-relevance, then it satisfies:
\begin{enumerate}
\item[] If $K \cro A \vdash \falsum$, then $K \subseteq (K \cro A)$ (\textbf{$\cro$-preservation}).
\end{enumerate}

\item If $\cro$ satisfies $\cro$-Relevance, then it satisfies $\cro$-preservation.

\end{enumerate}

\end{LEM}

\begin{proof}[Proof for Lemma \ref{lemma on interderivability between postulates}:]
\textit{1.} \textit{From left to right:} Let $A \cap K \neq \emptyset$, $B \cap K \neq \emptyset$ and $K \cup A = K \cup B$. Then, it follows that $A \subseteq A \cup B \subseteq K \cup A$ and $B \subseteq A \cup B \subseteq K \cup B$. So, by $\meo$-coincidence, $K \meo A = K\meo (A \cup B) = K \meo B$. \textit{From right to left:} Let $K \cap A \neq \emptyset$ and $A \subseteq B \subseteq K \cup A$. Then, it follows immediately that $K \cap B \neq \emptyset$ and $K \cup A = K \cup B$. So, by $\meo$-Coincidence, $K \meo A = K \meo B$.
\\
\textit{2.} Similar to the proof for the equivalence between $\meo$-coincidence and $\meo$-Coincidence.
\\
\textit{3.} The from left to right inclusion direction follows immediately from $\cro$-inclusion and the other inclusion direction follows immediately from $\cro$-relevance.
\\
\textit{4.} Let $K \cro A \vdash \falsum$, then $A \equiv \{\falsum\}$ by $\cro$-consistency. Suppose towards contradiction that $K \setminus (K \cro A) \neq \emptyset$, then, by $\cro$-relevance, there exists some $K^{\prime}$ such that $K^{\prime} \cup \{\varphi\} \nvdash \falsum$ for some $\varphi \in A$. It contradicts that $A \equiv \{\falsum\}$. Thus, $K \subseteq (K \cro A)$.
\\
\textit{5.} Let $K \cro A \vdash \falsum$. Suppose towards contradiction that $K \setminus (K \cro A) \neq \emptyset$, then, by $\cro$-Relevance, there exists some $K^{\prime}$ such that $K \cro A \subseteq K^{\prime}$ and $K^{\prime}  \nvdash \falsum$. It contradicts that $K \cro A \vdash \falsum$. Thus, $K \subseteq (K \cro A)$.
\end{proof}

\begin{proof}[Proof for Theorem \ref{Representation theorem for general expansion}:]
\textit{1.} \textit{From left to right:} Trivial. \textit{From right to left:} Let us define a selection function $\gamma$ in the following way:\\
(i) $\gamma (K, \mb{Y} ) = K$, if $\mb{Y}$ is empty;\\
(ii) $\gamma (K, \mb{Y} ) = \{X \in \mb{Y} \mid X \subseteq K \meo A \} $, if there is some non-empty $A$ such that $\mb{Y} = K \Join A$;\\
(iii) $\gamma (K, \mb{Y} ) =  \mb{Y}$, otherwise.
\\
We need to show that:\\
(1) $\gamma$ is well-defined, i.e. if $\mb{X} = \mb{Y}$, then $\gamma(K, \mb{X}) = \gamma (K, \mb{Y})$;
\\
(2) $\gamma(K, \mb{Y}) = K$ if $\mb{Y} = \emptyset$, which is immediate from the definition;
\\
(3) $\gamma(K, \mb{Y}) \subseteq \mb{Y}$ if $\mb{Y} \neq \emptyset$, which is also immediate from the definition;
\\
(4) $\gamma(K, \mb{Y}) \neq \emptyset$ if $\mb{Y} \neq \emptyset$; and
\\
(5) $K \meo A = \bigcup\gamma(K, K \Join A)$.
\\
For (1) we only need to show that if $K \Join A = K \Join B \neq \emptyset$, then $\gamma(K, K \Join A) = \gamma (K, K \Join B)$. If $A \cap K = \emptyset$, it is immediate from $K \Join A = K \Join B$ that $A = B$. So $K \meo A = K \meo B$, whence $\gamma(K, K \Join A) = \gamma (K, K \Join B)$. On the other hand, if $A \cap K \neq \emptyset$, then $K \in K \Join A$. So $K \in K \Join B$ and hence $B \cap K \neq \emptyset$. Moreover, it follows from $K \Join A = K \Join B$ that $K \cup A = K \cup B$. So, by $\meo$-coincidence and Lemma \ref{lemma on interderivability between postulates}, it holds that $K \meo A = K \meo B$, whence $\gamma(K, K \Join A) = \gamma (K, K \Join B)$.
\\
For (4), it is enough to show that $\gamma(K, K \Join A) \neq \emptyset$ if $A$ is non-empty. Let $A \neq \emptyset$, it follows immediately from $\meo$-inclusion, $\meo$-preservation and $\meo$-success that $K \meo A \in K \Join A$. So $\gamma(K, K \Join A) = \{X \in K \Join A \mid X \subseteq K \meo A \} \neq \emptyset $.
\\
As to (5), we first consider the case $A = \emptyset$. It follows from $A =\emptyset$ that  $K \Join A = \emptyset$. So, by the definition of $\gamma$, $\bigcup\gamma(K, K \Join A) = K$. Moreover, by $\meo$-inclusion and $\meo$-preservation, it follows from $A = \emptyset$ that $K \meo A =K$. So $\bigcup\gamma(K, K \Join A) = K \meo A$. For the principle case $A \neq \emptyset$, as we have shown that $K \meo A \in K \Join A$ when $A \neq \emptyset$, it follows immediately that $\bigcup\gamma(K, K \Join A)  = \bigcup \{X \in K \Join A \mid X \subseteq K \meo A \} = K \meo A $.
\\
\\
\textit{2.} \textit{From left to right:} Trivial. \textit{From right to left:} Define the selection function $\gamma$ in the same way as that in the proof for the first part. We omit checking the details.
\end{proof}

\begin{proof}[Proof for Theorem \ref{representation theorem for internal choice revision}:]
\textit{1. From left to right:} Let $\cro$ be an internal choice revision on consistent belief bases with finite sets as inputs, i.e. there exist a consistency-preserving partial expansion $\meo$ and a package contraction $\pco$ such that  for every consistent $K$ and finite  $A$, $K \cro A = K \cco n(A) \meo A.$
\\
\textit{$\cro$-inclusion:} By $\meo$-inclusion, $ K \cco n(A) \meo A \subseteq  (K \cco n(A)) \cup A$. Moreover, $K \cco n(A) \subseteq K $ by $\cco$-inclusion. Thus, $K \cro A = K \cco n(A) \meo A  \subseteq K \cup A$.
\\
\textit{$\cro$-success:} It follows immediately from $\meo$-success and the definition of internal choice revision.
\\
\textit{$\cro$-iteration:}  It is easy to see that $(K \cco A) \cco A = K \cco A$ for every $A$. So, $K \cro A = K \cco n(A) \meo A = (K \cco n(A)) \cco n(A) \meo A$. By Observation \ref{observation for internal choice revision}, $ (K \cco n(A)) \cco n(A) \meo A=  (K \cap (K \cro A)) \cco A \meo A$. Thus, $K \cro A = (K \cap (K \cro A)) \cco A \meo A= (K \cap (K \cro A)) \cro A$.
\\ 
\textit{$\cro$-consistency:} Since $K$ is consistent, by the definition of $\cco$, $K \cco n(A)$ is consistent for every $A$. Given that $K \cco n(A)$ is consistent and $\meo$ is consistency-preserving,  it is easy to see that $K \cco n(A) \meo A \vdash \falsum$ only if $A \equiv \{\perp\}$.  
\\
\textit{$\cro$-coincidence:} Let $A \cap K \neq \emptyset$, $B \cap K \neq \emptyset$ and $A \cup K = B \cup K$. It follows from  $A \cap K \neq \emptyset$, $B \cap K \neq \emptyset$ and $K$ is consistent that $K \cco n(A) = K = K \cco n(B)$. By Lemma \ref{lemma on interderivability between postulates}, $\meo$ satisfies $\meo$-Coincidence. So, $K \cco n(A) \meo A = K \cco n(B) \meo B$, i.e. $K \cro A = K \cro B$.
\\
\textit{$\cro$-uniformity:} Assume $(\star)$ that it holds for all $K^{\prime} \subseteq K$ that $K^{\prime} \cup \{\varphi\} \vdash \falsum$ for every $\varphi \in A$ iff $K^{\prime} \cup \{\psi\} \vdash \falsum$ for every $\psi \in B$. By Observation \ref{observation for internal choice revision}, $ K \cap (K \cro A) = K \cco n(A)$ for every $A$. So, in order to complete the proof, we only need demonstrate that $K \angle n(A) = K \angle n(B)$. We need consider two cases. (i) $A$ is empty or $B$ is empty. Without loss of generality, let $A = \emptyset$. Then, it holds vacuously that $\emptyset \cup \{\varphi\} \vdash \falsum$ for every $\varphi \in A$. So, by $(\star)$, either $B$ is empty or $B \equiv \{\perp\}$, whence $n(B) \equiv \{\taut\} $, i.e. $n(B) \equiv n(A)$. Hence, $K \angle n(A) = K \angle n(B)$. (ii) $A \neq\emptyset$ and $B \neq \emptyset$. Then, it follows from $(\star)$ that for all $K^{\prime} \subseteq K$, $K^{\prime} \Vdash n(A)$ iff $K^{\prime} \Vdash n(B)$. So, $K \angle n(A) = K \angle n(B)$.
\\
\textit{$\cro$-relevance:} Let $\varphi \in K \setminus K \cro A$, then $K \cco n(A) \neq K$, i.e. $K  \angle n(A) \neq \emptyset$. So there exists $X \in K  \angle n(A)$ such that $K \cco A \subseteq X$ and $\varphi \notin X$.  By Observation \ref{observation for internal choice revision}, $K \cap (K \cro A) = K \cco A$, so it follows from $K \cco A \subseteq X$ that $K \cap (K \cro A)  \subseteq X$. Moreover, it follows from $X \in K \angle n(A)$ and $\varphi \in K \setminus X$ that $X \nVdash n(A)$ and $X \cup \{\varphi\} \Vdash n(A)$, i.e. $X \cup \{\psi\}\nvdash \falsum$ for some $\psi \in A$ and $X \cup \{\varphi\} \cup \{\lambda\} \vdash \falsum$ for every $\lambda \in A$. Thus, $X$ is the set we are looking for.
\\
\\
\textit{From right to left:} 
Let $\cro$ be some operator satisfying the proposed postulates. Let us define a selection function $\gamma$ in the following way:\\
(i) $\gamma(X, \mb{Y}) = X$ if $\mb{Y} = \emptyset$.\\
(ii) If there exist consistent $K$ and finite $A$ such that $X = K $ and $\mb{Y} = K \angle n(A) \linebreak \neq \emptyset$, then $\gamma(X, \mb{Y}) = \{Y \in \mb{Y} \mid K \cap (K \cro A) \subseteq Y \}$.\\
(iii) Otherwise, $\gamma(X, \mb{Y}) = \mb{Y}$.
\\
We need to show that:\\
(1) $\gamma$ is well-defined, i.e. if $\mb{X} = \mb{Y}$, then $\gamma(X, \mb{X}) = \gamma(X, \mb{Y})$;\\
(2) $\gamma(X, \mb{Y}) = X$ if $\mb{Y} = \emptyset$, which is immediate from the definition;\\
(3) $\gamma(X, \mb{Y}) \subseteq \mb{Y}$ if $\mb{Y} \neq \emptyset$, which is also immediate from the definition; and\\
(4) $\gamma(X, \mb{Y}) \neq \emptyset$ if $\mb{Y} \neq \emptyset$.\\
For (1), let $K$ be consistent and both $A$ and $B$ be finite sets, and let $X = K$ and let $\mb{X} = K \angle n(A) = K \angle n(B) = \mb{Y} \neq \emptyset$. Then $A \neq \emptyset$, $B \neq \emptyset$. Moreover, by Observation \ref{observation on partial upper bound property of choice remainder set}, it follows that for all $K^{\prime} \subseteq K$, $K^{\prime} \nVdash n(A)$ iff $X^{\prime} \nVdash n(B)$, i.e.  $K^{\prime} \cup \{\varphi\} \vdash \falsum$ for every $\varphi \in A$ iff $K^{\prime} \cup \{\psi\} \vdash \falsum$ for every $\psi \in B$. So, by $\cro$-uniformity, $K \cap (K \cro A) = K \cap (K \cro B)$. By the definition of $\gamma$, it follows immediately that $\gamma(K, K \angle A) = \gamma(K, K \angle B)$.
\\
For (4), let $K$ be consistent and $A$ be finite, and let $K \angle n(A) \neq \emptyset$, we need to show that $\gamma(K, K \angle n(A)) \neq \emptyset$. It follows from $K \angle n(A) \neq \emptyset$ that $A \neq \emptyset$ and $A \not \equiv \{\perp\}$. So, by $\cro$-success$, K \cro A \vdash A$ and, by $\cro$-consistency, $K \cro A$ is consistent. It follows that $K \cro A \nVdash n(A)$ and hence $K \cap (K \cro A) \nVdash n(A)$. So, by Observation \ref{observation on partial upper bound property of choice remainder set}, there exists $K^{\prime}$ such that $K \cap (K \cro A) \subseteq K^{\prime} \subseteq K$ and $K^{\prime} \in K \angle n(A)$, i.e. $\gamma(K, K \angle n(A)) \neq \emptyset$.
\\
Furthermore, let us define another selection function $\gammap$ in the following way:\\
(i) $\gamma^{\prime}(X, \mb{Y}) = X$ if $\mb{Y} = \emptyset$.\\
(ii) If there exist a consistent $K$ and a finite $A$ such that $X = K \cap (K \cro A)$ and $\mb{Y} = X \Join A \neq \emptyset$, then $\gamma^{\prime}(X, \mb{Y}) = \{Y \in \mb{Y} \mid Y \subseteq K \cro A \}$.\\
(iii) If $\{Y \in \mb{Y} \mid Y \mbox{ is consistent}\} \neq \emptyset$ and there exists no consistent $K$ and  finite $A$ such that $X = K \cap (K \cro A)$ and $\mb{Y} = X \Join A \neq \emptyset$, then $\gamma^{\prime}(X, \mb{Y}) \in \{Y \in \mb{Y} \mid Y \mbox{ is consistent}\}$.\\
(iv) Otherwise, $\gamma^{\prime}(X, \mb{Y}) = \mb{Y}$.
\\
We need show that:\\
($1^{\prime}$) $\gamma^{\prime}$ is well-defined, i.e. if $\mb{X} = \mb{Y}$, then $\gammap (X, \mb{X}) = \gammap (X, \mb{Y})$;\\
($2^{\prime}$) $\gammap (X, \mb{Y}) = K$ if $\mb{Y} = \emptyset$, which is immediate from the definition;\\
($3^{\prime}$) $\gammap (X, \mb{Y}) \subseteq \mb{Y}$ if $\mb{Y} \neq \emptyset$, which is also immediate from the definition;\\
($4^{\prime}$) $\gammap (X, \mb{Y}) \neq \emptyset$ if $\mb{Y} \neq \emptyset$; and\\
($5^{\prime}$) $\gammap$ is $\Join$-consistency-preserved, i.e. $\bigcup \gammap (X, \mb{Y})$ is consistent, when there exists a set $A$ such that $\mb{Y} = X \Join A$ and $\{Y \in \mb{Y} \mid Y \mbox{ is consistent}\} \neq \emptyset$.\\ 
For $(1^{\prime})$, let $K_1$ and $K_2$ be consistent, let $A$ and $B$ be finite, and let $K = K_1 \cap (K_1 \cro A) = K_2 \cap (K_2 \cro B)$ and $K \Join A = K \Join B \neq \emptyset$, in order to prove $\gammap (K, K_1 \cap (K_1 \cro A)) = \gammap (K, K_2 \cap (K_2 \cro B))$, by the definition of $\gammap$, we only need to show that $K_1 \cro A = K_2 \cro B$. There are two cases. (i) $A \cap K = \emptyset$. Then, it follows from $K \Join A = K \Join B $ that $B \cap K =\emptyset$. By Lemma \ref{lemma on interderivability between postulates}, $\cro$ satisfies $\cro$-vacuity, whence it follows that $K \cro A =K = K \cro B$. Hence, by $\cro$-iteration, $K_1 \cro A = K \cro A =  K \cro B = K_2 \cro B$. (ii) $A \cap K \neq \emptyset$. Then it follows from $K \Join A = K \Join B \neq \emptyset$ that $B \cap K \neq \emptyset$ and $A \cup K = B \cup K$. By Lemma \ref{lemma on interderivability between postulates}, $\cro$ satisfies $\cro$-Coincidence. So, $K \cro A = K \cro B$ and hence, by $\cro$-iteration, $K_1 \cro A = K_2 \cro B$.
\\
For $(4^{\prime})$, let $K$ be consistent and $A$ be finite, and let $(K \cap (K \cro A)) \Join A \neq \emptyset$, in order to complete the proof, it is enough to show that $K \cro A \in (K \cap (K \cro A)) \Join A$. So we only need to show (a) $K \cro A \subseteq A \cup (K \cap (K \cro A))$ which is immediate from $\cro$-inclusion; (b) $(K \cap (K \cro A)) \subseteq K \cro A$ which is trivial; and (c) $A \cap (K \cro A) \neq \emptyset$. It follows from $(K \cap (K \cro A)) \Join A \neq \emptyset$ that $A$ is non-empty. So, (c) follows immediately from $\cro$-success. 
\\
As to $(5^{\prime})$, let $K$ be consistent and $A$ be finite, and let $(K \cap (K \cro A) ) \Join A \neq \emptyset$. Moreover, assume that $(\dagger)$ there exists a consistent $X \in (K \cap (K \cro A) ) \Join A$. It has been shown in the proof for $(4^{\prime})$ that it follows that $K \cro A \in (K \cap (K \cro A) ) \Join A$. In order to complete the proof, we only need to show that $K \cro A$ is consistent. It follows immediately from $(\dagger)$ that $A \not \equiv \{\perp\}$. So, by $\cro$-consistency, $K \cro A$ is consistent.
\\
\\
To finish the proof,  we need to show that for every consistent $K$ and finite $A$,
\\
$(\star)$ $K \cro A = \bigcup\gamma^{\prime}(\bigcap \gamma(K, K \angle n(A)), \bigcap \gamma(K, K \angle n(A)) \Join A))$ holds. 
\\
We first prove that $\bigcap \gamma(K, K \angle n(A)) = K \cap (K \cro A)$. Let us consider two cases according to whether $K \angle n(A)$ is empty. (i) $K \angle n(A) = \emptyset$. Then, by the definition of $\gamma$, $\bigcap \gamma(K, K \angle n(A)) = K$. Moreover, it follows from $K \angle n(A)$ that $A = \emptyset$ or $A \equiv \{\perp\}$. If $A = \emptyset$, then $K \cro A =K$ by $\cro$-vacuity which $\cro$ satisfies according to Lemma \ref{lemma on interderivability between postulates}. If $A \equiv \{\perp\}$, then $K \cro A \vdash \falsum$ by $\cro$-success. Moreover, by Lemma \ref{lemma on interderivability between postulates}, $\cro$ satisfies $\cro$-preservation. It follows that $K \subseteq K \cro A$. So in each case it follows that $K = K \cap (K \cro A)$, i.e. $\bigcap \gamma(K, K \angle n(A)) = K \cap (K \cro A)$. (ii) $K \angle n(A) \neq \emptyset$. Then, the from right to left inclusion direction follows immediately from the definition of $\gamma$. For the other inclusion direction, suppose that $(\ddagger)$ there exists $\varphi \notin K \cap (K \cro A)$ such that $\varphi \in X$ for every $X \in \gamma(K,  K \angle n(A))$. It follows immediately that $\varphi \in K \setminus K \cro A$. So, by $\cro$-relevance, there is $K^{\prime}$ such that $K \cap (K \cro A) \subseteq K^{\prime} \subseteq K $, $K^{\prime} \cup \{\psi\} \nvdash \falsum$ for some $\psi \in A$ and $K^{\prime} \cup \{\varphi\} \cup \{\lambda\} \vdash \falsum$ for every $\lambda \in  A$. Note that $A \neq \emptyset$ since $K \angle n(A) \neq \emptyset$. So, it follows that $K^{\prime} \nVdash n(A)$ and $K^{\prime} \cup \{\varphi\} \Vdash n(A)$. So, by Observation \ref{observation on partial upper bound property of choice remainder set}, there exists $Y$ such that $ K^{\prime } \subseteq Y \subseteq K$ and $Y \in K \angle n(A)$. Also, it follows from $K \cap (K \cro A) \subseteq K^{\prime} \subseteq Y$ that $Y \in \gamma(K, K \angle n(A))$. Hence, by $(\ddagger)$, $\varphi \in Y$. However, it follows from $K^{\prime} \cup \{\varphi\} \Vdash n(A)$ and $K^{\prime}\cup \{\varphi\} \subseteq Y$ that $Y \Vdash n(A)$ which contradicts that $Y \in K \angle n(A)$. Thus, there is no such kind of $\varphi$, i.e. $\bigcap \gamma(K, K \angle n(A)) \subseteq K \cap (K \cro A)$. So $(\star)$ holds iff\\
$(\star \star)$ $K \cro A = \bigcup\gamma^{\prime}(K \cap (K \cro A), K \cap (K \cro A) \Join A))$ holds.\\
As to $(\star \star)$, we consider two cases. (i) $A = \emptyset$. Then, by $\cro$-vacuity, $K \cro A = K$ and $K \cap (K \cro A) \Join A = \emptyset$. So $(\star \star)$ is immediate from the definition of $\gammap$. (ii) $A \neq \emptyset$. Then, $K \cap (K \cro A) \Join A \neq \emptyset$. In this case, as we have shown in the proof for $(4^{\prime})$, $K \cro A \in K \cap (K \cro A) \Join A  $. So, by the definition of $\gammap$, it holds that $\bigcup\gamma^{\prime}(K \cap (K \cro A), K \cap (K \cro A) \Join A)) = \bigcup\{ X \in K \cap (K \cro A) \Join A) \mid X \subseteq K \cro A \} = K \cro A$.  
\\
\\
\textit{2. From left to right:} We only need to check $\cro$-redundancy: Assume $K \cup Z$ is consistent, $A \not \equiv \{\perp\}$, $A \neq \emptyset$ and it holds for every $\varphi \in Z$ that $\varphi \vdash \neg \psi$ for every $\psi \in A$. We first show that $K \angle n(A) = (K \cup Z) \angle n(A)$. Since $A \not \equiv \{\perp\}$ and $A \neq \emptyset$, it follows that $K \angle n(A) \neq \emptyset$ and $ (K \cup Z) \angle n(A) \neq \emptyset$. Let $X \in K \angle n(A)$. Then, $X \nVdash n(A)$ and $X \subseteq K \subseteq (K \cup Z)$. Since it holds for every $\varphi \in Z$ that $\varphi \vdash \neg \psi$ for every $\psi \in A$, it follows that $X \cup \{\varphi\} \Vdash n(A)$ for every $\varphi \in Z$. Moreover, it follows from $X \in K \angle n(A)$ that $X \cup \{\varphi\} \Vdash n(A)$ for every $\varphi \in K \setminus X$. So $X \cup \{\varphi\} \Vdash n(A)$ for every $\varphi \in (K \cup Z) \setminus X$. It follows that $X \in (K \cup Z) \angle n(A)$. For the other inclusion direction, let $X \in (K \cup Z) \angle n(A)$. In order to show that $X \in K \angle n(A)$, it is enough to prove that $X \subseteq K$, which follows directly from that it holds for every $\varphi \in Z$ that $\varphi \vdash \neg \psi$ for every $\psi \in A$. So, $K \angle n(A) = (K \cup Z) \angle n(A) \neq \emptyset$. Hence, by Observation \ref{observation for internal choice revision}, it follows that $K \cap (K \cro A) = (K \cup Z) \cap ((K \cup Z) \cro A)$. So, by $\cro$-iteration, it follows that $K \cro A = (K \cup Z) \cro A$. 
\\
\\
\textit{From right to left:} Let $\gamma$ be constructed in the same way as the $\gamma$ in the proof of the first part. By observation \ref{Observation that all selection functions are Join-unified}, in order to complete the proof, it is enough to show that $\gamma$ is $\angle$-unified. So, let us suppose $(\star)$ $K_1 \angle n(A) = K_2 \angle n(B) \neq \emptyset$, we will prove that $K_1 \cap (K_1 \cro A) = K_2 \cap (K_2 \cro B)$. By observation \ref{observation on partial upper bound property of choice remainder set}, it holds for every $\varphi \in K_1 \setminus \bigcup (K_1 \angle n(A))$ that $\varphi \Vdash n(A)$. Moreover, it follows from $K_1 \angle n(A) \neq \emptyset$ that $A \not \equiv \{\perp\}$ and $A \neq \emptyset$. Hence, by $\cro$-redundancy, $K_1 \cro A = (\bigcup (K_1 \angle n(A)) \cro A$. By $\cro$-consistency and $\cro$-success, it follows from $A \not \equiv \{\perp\}$ and $A \neq \emptyset$ that $K_1 \cap (K_1 \cro A) \nVdash n(A)$. So, by Observation \ref{observation on partial upper bound property of choice remainder set}, $K_1 \cap (K_1 \cro A) \subseteq \bigcup (K_1 \angle n(A))$. It follows from this and $K_1 \cro A = (\bigcup (K_1 \angle n(A)) \cro A$ that (a) $K_1 \cap (K_1 \cro A) = (\bigcup (K_1 \angle n(A)) \cap ((\bigcup (K_1 \angle n(A)) \cro A)$. It can be obtained by a similar argument that (b) $K_2 \cap (K_2 \cro B) =(\bigcup (K_2 \angle n(B)) \cap ((\bigcup (K_2 \angle n(B)) \cro B)$. Let $X $ be any subset of $ \bigcup (K_1 \angle n(A)$. By Observation \ref{observation on partial upper bound property of choice remainder set}, $X \nVdash n(A)$ iff $X \subseteq Y$ for some $Y \in K_1 \angle n(A)$ iff $X \subseteq Y$ for some $Y \in K_2 \angle n(B) $ by $(\star)$ iff $X \nVdash n(B)$. So, by $(\star)$ and $\cro$-uniformity, $(\bigcup (K_1 \angle n(A)) \cro A = (\bigcup (K_2 \angle n(B)) \cro B$ and hence (c) $(\bigcup (K_1 \angle n(A)) \cap ((\bigcup (K_1 \angle n(A)) \cro A) = (\bigcup (K_2 \angle n(B)) \cap ((\bigcup (K_2 \angle n(B)) \cro B)$. It follows from (a), (b) and (c) that $K_1 \cap (K_1 \cro A) = K_2 \cap (K_2 \cro B)$.
\end{proof}

\begin{proof}[Proof for Theorem \ref{representation theorem for external choice revision}:]
\textit{1. From left to right:} Let $\cro$ be an external choice revision, i.e. there exist a partial expansion $\meo$ and a package contraction $\pco$ such that  for every sets $K$ and $A$, $K \cro A = K \meo A \pco n(A^{\prime})$, where $A^{\prime} = (K \meo A) \setminus K$.
\\
\textit{$\cro$-inclusion:} By the definition of $\cro$ and $\pco$-inclusion, $K \cro A \subseteq K \meo A$. Moreover, $K \meo A \subseteq K \cup A$ by $\meo$-inclusion. Thus, $K \cro A \subseteq K \cup A$.
\\
\textit{$\cro$-success:} Let $A \neq \emptyset$, then $A \cap (K \meo A) \neq \emptyset$ by $\meo$-success. If $A^{\prime} = \emptyset$, then $K \cro A = (K \meo A) \pco \{\taut \} = K \meo A$. It follows that $A \cap (K \cro A) \neq \emptyset$. If $A^{\prime} \neq \emptyset$, by Observation \ref{lemma for external choice revision}, $\Ap \subseteq K \cro A$. Moreover, $\Ap \subseteq A$ by $\meo$-inclusion. It follows that $A \cap (K \cro A) \neq \emptyset$ holds.
\\
\textit{$\cro$-confirmation:} Let $(A \cap (K \cro A) \subseteq K$.    Then, $A \cap (K \meo A) \subseteq K$. It follows that $K \meo A = K$. Thus, $K \cro A = K \meo A \pco A^{\prime} = K \pco \{\taut\} = K$.  
\\
\textit{$\cro$-Consistency:} Let $(K \cro A) \setminus K$ is non-empty and consistent. By Observation \ref{lemma for external choice revision}, $\Ap = (K \meo A) \setminus K $ is non-empty and consistent. It follows that $(K \meo A) \upVdash n(\Ap) \neq \emptyset$. So $ K \meo A \pco n(\Ap) \nvdash \falsum$, i.e. $K \cro A$ is consistent. 
\\
\textit{$\cro$-coincidence:} Let $K \cap A \neq \emptyset$, $K \cap B \neq \emptyset$ and $A \setminus K = B \setminus K$. By Lemma \ref{lemma on interderivability between postulates}, $\meo$ satisfies $\meo$-Coincidence. So, $K \meo A = K \meo B$. So $(K \meo A) \setminus K = (K \meo B) \setminus K$. It follows that $K \cro A = K \cro B$ by the definition of $\cro$.
\\
\textit{$\cro$-Uniformity:} Assume that $(\dagger)$ $K_1 \neq ((K_1 \cro A) \cup K_1 ) =K = ((K_2 \cro B) \cup K_2) \neq K_2$ and that $(\ddagger)$ it holds for all $K^{\prime} \subseteq K$ that $K^{\prime} \cup ((K_1 \cro A) \setminus K_1) \vdash \falsum$  iff $K^{\prime} \cup  ((K_2 \cro B) \setminus K_2) \vdash \falsum$. By Observation \ref{lemma for external choice revision}, it follows from $(\dagger)$ that $K = \Kone \meo A = \Ktwo \meo B$. So, in order to prove that $\Kone \cro A = \Ktwo \cro B$, we only need to show that $K \upVdash n(\Ap) = K \upVdash n(\Bp)$ where $\Ap = (\Kone \meo A) \setminus K_1$ and $\Bp = (K_2 \meo B) \setminus K_2$. By Observation \ref{lemma for external choice revision}, it follows from $(\ddagger)$ that for all $K^{\prime} \subseteq K$, $K^{\prime} \cup A^{\prime} \vdash \falsum$  iff $K^{\prime} \cup \Bp \vdash \falsum$. Moreover, it follows from $K_1 \neq ((K_1 \cro A) \cup K_1 ) $ and $ ((K_2 \cro B) \cup K_2) \neq K_2$ that both $\Ap$ and $\Bp$ are non-empty. So, it follows that for all $K^{\prime} \subseteq K$, $K^{\prime} \vdash n(\Ap)$ iff $K^{\prime} \vdash n(\Bp)$. So, $K \upVdash n(\Ap) = K \upVdash n(\Bp)$.
\\
\textit{$\cro$-Relevance:} Let $\varphi \in K \setminus K \cro A$, then $K \meo A \upVdash n(A^{\prime}) \neq \emptyset$. So there exists $X \in K \meo A \upVdash n(A^{\prime})$ such that $K \cro A \subseteq X$ and $\varphi \notin X$. We need to show that (i) $X \subseteq K \cup K \cro A$, (ii) $X \nvdash \falsum$ and (iii) $X \cup \{\varphi\} \vdash \falsum$. (ii) is immediate from the definition of package remainder set. For (i) we only need to prove that $K \cup (K \cro A)  = K \meo A$, which is part of Observation \ref{lemma for external choice revision}. For (iii), we first show that $\varphi \in K \meo A$, which follows immediately from $\varphi \in K$ by $\meo$-preservation. So, by the definition of package remainder set, $X \cup \{\varphi\} \vdash n(A^{\prime})$. By Observation \ref{lemma for external choice revision}, $A^{\prime} \subseteq K \cro A$. So, it follows from $K \cro A \subseteq X$ that $X \Vdash A^{\prime}$. Moreover, since $(K \meo A) \upVdash n(A^{\prime})$ is not empty, it follows that $A^{\prime} \neq \emptyset$. So, $X \cup \{\varphi\}$ is inconsistent.  
\\
\\
\textit{From right to left:} Let us define a selection function $\gamma$ in the following way:\\
(a) $\gamma(K, \mb{Y}) = K$ if $\mb{Y} = \emptyset$.\\
(b) If there exists $A$ such that $\mb{Y} = K \Join A \neq \emptyset$, then $\gamma(K, \mb{Y}) = \{Y \in \mb{Y} \mid Y \subseteq K \cup (K \cro A) \}$.\\
(c) Otherwise, $\gamma(K, \mb{Y}) = \mb{Y}$.
\\
We need to show that:\\
(1) $\gamma$ is well-defined, i.e. if $\mb{X} = \mb{Y}$, then $\gamma(K, \mb{X}) = \gamma (K, \mb{Y})$;\\
(2) $\gamma(K, \mb{Y}) = K$ if $\mb{Y} = \emptyset$, which is immediate from the definition;\\
(3) $\gamma(K, \mb{Y}) \subseteq \mb{Y}$ if $\mb{Y} \neq \emptyset$, which is also immediate from the definition;\\
(4) $\gamma(K, \mb{Y}) \neq \emptyset$ if $\mb{Y} \neq \emptyset$;\\
For (1) we only need to show that if $K \Join A = K \Join B \neq \emptyset$, then $\gamma(K, K \Join A) = \gamma (K, K \Join B)$. If $A \cap K = \emptyset$, it is immediate from $K \Join A = K \Join B$ that $A = B$. So $K \cup (K \cro A) = K \cup (K \cro B)$, whence $\gamma(K, K \Join A) = \gamma (K, K \Join B)$. On the other hand, if $A \cap K \neq \emptyset$, then $K \in K \Join A$. So $K \in K \Join B$ and hence $B \cap K \neq \emptyset$. Moreover, it follows from $K \Join A = K \Join B$ that $K \cup A = K \cup B$. By Lemma \ref{lemma on interderivability between postulates}, $\cro$ satisfies $K \cro A \subseteq X$. So, it holds that $K \cro A = K \cro B$, i.e. $K \cup (K \cro A) = K \cup (K \cro B)$, whence $\gamma(K, K \Join A) = \gamma (K, K \Join B)$ as well.
\\
For (4), it is enough to show that $K \cup (K \cro A) \in K \Join A$ when $K \Join A$ is non-empty. We need to show (a) $K \subseteq K \cup (K \cro A)$ which is trivial; (b) $K \cup (K \cro A) \subseteq K \cup A$ which is immediate from $\cro$-inclusion; and (c) $A \cap (K \cup (K \cro A)) \neq \emptyset$. It follows from $K \Join A \neq \emptyset$ that $A \neq \emptyset$. So, by $\cro$-success, $A \cap (K \cro A) \neq \emptyset$, whence (c) follows immediately. 
\\
Furthermore, let us define another selection function $\gamma^{\prime}$ in the following way:\\
(a) $\gamma^{\prime}(X, \mb{Y}) = X$ if $\mb{Y} = \emptyset$.\\
(b) If there exist $K$ and $A$ such that $X = K \cup (K \cro A)$ and $\mb{Y} = X \upVdash n(A^{\prime}) \neq \emptyset$ where $A^{\prime} = (K \cro A) \setminus K$, then $\gamma^{\prime}(X, \mb{Y}) = \{Y \in \mb{Y} \mid K \cro A \subseteq Y \}$.\\
(c) Otherwise, $\gamma^{\prime}(X, \mb{Y}) = \mb{Y}$.
\\
We need to show that:\\
$(1^{\prime})$ $\gamma^{\prime}$ is well-defined, i.e. if $\mb{X} = \mb{Y}$, then $\gamma^{\prime}(K, \mb{X}) = \gamma^{\prime} (K, \mb{Y})$;\\
$(2^{\prime})$ $\gamma^{\prime}(K, \mb{Y}) = K$ if $\mb{Y} = \emptyset$, which is immediate from the definition;\\
$(3^{\prime})$ $\gamma^{\prime}(K, \mb{Y}) \subseteq \mb{Y}$ if $\mb{Y} \neq \emptyset$, which is also immediate from the definition;\\
$(4^{\prime})$ $\gamma^{\prime}(K, \mb{Y}) \neq \emptyset$ if $\mb{Y} \neq \emptyset$;\\
For $(1^{\prime})$, we only consider the principle case, i.e. the case of that $X = (K_1 \cro A) \cup K_1  = (K_2 \cro B) \cup K_2$ and $\mb{Y} = X \upVdash n(A^{\prime}) = X \upVdash n(B^{\prime}) \neq \emptyset $ where $A^{\prime} = (K_1 \cro A) \setminus K_1$ and $B^{\prime} = (K_2 \cro B) \setminus K_2$. It follows from $X \upVdash n(A^{\prime}) = X \upVdash n(B^{\prime}) \neq \emptyset $ that  $\Ap \neq \emptyset$ and $\Bp \neq \emptyset$. Moreover, by Observation \ref{observation on upper bound property}, it follows from $X \upVdash n(A^{\prime}) = X \upVdash n(B^{\prime}) $ that for all $X^{\prime} \subseteq X$, $X^{\prime} \nvdash n(A^{\prime})$ iff $X^{\prime} \nvdash n(B^{\prime})$. So, $K_1 \neq ((K_1 \cro A) \cup K_1 ) =X = ((K_2 \cro B) \cup K_2) \neq K_2$ and it holds for all $X^{\prime} \subseteq X$ that $X^{\prime} \cup ((K_1 \cro A) \setminus K_1) \nvdash \falsum$  iff $X^{\prime} \cup  ((K_2 \cro B) \setminus K_2) \nvdash \falsum$. Hence, by $\cro$-Uniformity, $K_1 \cro A = K_2 \cro B$. Thus, $\gamma^{\prime}(X, X \upVdash A^{\prime}) = \gamma^{\prime}(X, X \upVdash B^{\prime})$.
\\
For $(4^{\prime})$, we only need to show that $\gamma^{\prime}(K \cup (K \cro A), \mb{Y}) \neq \emptyset$ when $\mb{Y} = ((K \cup (K \cro A)) \upVdash n(A^{\prime})) \neq \emptyset$ where $A^{\prime} = (K \cro A) \setminus K$. We first show that $K \cro A \nvdash n(A^{\prime})$. It follows immediately from $((K \cup (K \cro A)) \upVdash n(A^{\prime})) \neq \emptyset$ that $A^{\prime}$ is consistent and non-empty. So, by $\cro$-Consistency, $K \cro A$ is consistent. Moreover, $A^{\prime} = ((K \cro A) \setminus K) \subseteq K \cro A$, i.e. $K \cro A \Vdash A^{\prime}$. Hence, $K \cro A \nvdash n(A^{\prime})$. So, by Observation \ref{observation on upper bound property}, there exists $Y$ such that $K \cro A \subseteq Y \subseteq K\cup (K \cro A)$ and $Y \in \mb{Y}$.
\\
\\
To finish the proof, we will prove that
\\ 
$(\star)$ $K \cro A = \bigcap\gamma^{\prime}(\bigcup \gamma(K, K \Join A), (\bigcup \gamma(K, K \Join A)) \upVdash n(A^{\prime}))$, where $A^{\prime} = (\bigcup \gamma(K, K \Join A)) \setminus K$. 
\\
Let us consider two cases according to whether $A$ is empty. In the first case, $A = \emptyset$. Then it follows immediately from $\cro$-confirmation that $K \cro A = K$. Moreover, by the definitions of those operations in $(\star)$, it is easy to see that  $K = \bigcap\gamma^{\prime}(\bigcup \gamma(K, K \Join A), (\bigcup \gamma(K, K \Join A)) \upVdash n(A^{\prime}))$ when $A = \emptyset$. So $(\star)$ holds in this case. In the remaining case, $A \neq \emptyset$. Then, $K \Join A \neq \emptyset$. As we have shown in the proof for $(4)$, $K \cup (K \cro A) \in K \Join A$ when $K \Join A$ is non-empty. So, by the definition of $\gamma$, $(\star)$ holds iff
\\
$(\star \star)$ $K \cro A = \bigcap\gamma^{\prime}(K\cup(K \cro A), (K\cup(K \cro A)) \upVdash n(A^{\prime \prime}))$ holds, where $A^{\prime \prime} = (K \cro A) \setminus K $.
\\
As to $(\star \star)$, we also need to consider two cases, according to whether $(K\cup(K \cro A)) \upVdash n(A^{\prime \prime}))$ is empty. In the first case, $(K\cup(K \cro A)) \upVdash n(A^{\prime \prime})) = \emptyset$. Then, it is easy to see that either (i) $n(A^{\prime \prime}) = \emptyset$ or (ii)$A^{\prime \prime} \vdash \falsum$. Suppose (i), it follows immediately that $K \cro A \subseteq K$. So, by $\cro$-confirmation, it follows that $K \cro A = K$, whence $(\star \star)$ follows immediately. Suppose (ii), it follows immediately that $K \cro A$ is inconsistent. By Lemma \ref{lemma on interderivability between postulates}, $\cro$ satisfies $\cro$-preservation. So, $K \subseteq (K \cro A)$, i.e. $K \cup (K \cro A) = K \cro A$.  Hence, $\bigcap\gamma^{\prime}(K\cup(K \cro A), (K\cup(K \cro A)) \upVdash n(A^{\prime \prime})) = \bigcap\gamma^{\prime}(K \cro A, \emptyset ) = K \cro A$, i.e. $(\star \star)$ holds. 
\\
Let us proceed the proof by considering the remaining case that $(K\cup(K \cro A)) \upVdash n(A^{\prime \prime}))$ is non-empty. In this case, the from left to right inclusion direction follows immediately from the clause (b) of the definition of $\gamma^{\prime}$. For the other inclusion direction, let us suppose towards contradiction that $(\dagger)$ there exists $\varphi \in (K \cup (K \cro A))$ such that $\varphi \notin K \cro A$ but $\varphi \in \bigcap \gamma^{\prime}(K\cup(K \cro A), (K\cup(K \cro A)) \upVdash n(A^{\prime \prime}))$. It follows from $\varphi \in (K \cup (K \cro A))$ and $\varphi \notin K \cro A$  that $\varphi \in K \setminus K \cro A$. So, by $\cro$-Relevance, there is some $K^{\prime}$ with $K \cro A \subseteq K^{\prime} \subseteq K \cup (K \cro A)$, such that $K^{\prime} \nvdash \falsum$ and $K^{\prime} \cup \{\varphi\} \vdash \falsum$. It follows from $K \cro A \subseteq K^{\prime}$ that $K^{\prime} \Vdash A^{\prime \prime}$. Since $A^{\prime \prime} \neq \emptyset$ and $K^{\prime} $ is consistent, it follows that $K^{\prime} \nvdash n(A^{\prime \prime})$. So, by Observation \ref{observation on upper bound property}, there exists $X$ such that $K \cro A \subseteq K^{\prime} \subseteq X$ and  $X \in ((K\cup(K \cro A)) \upVdash n(A^{\prime \prime}))$. By the definition of $\gamma^{\prime}$, it follows that $X \in \gamma^{\prime}(K\cup(K \cro A), (K\cup(K \cro A)) \upVdash n(A^{\prime \prime}))$. So, by $(\dagger)$, $\varphi \in X$. However, it follows from $K^{\prime} \cup \{\varphi\} \vdash \perp$ and $K^{\prime} \subseteq X$ that $X \cup \{\varphi\} \vdash \perp$, i.e. $X \vdash \falsum$, which contradicts that $X$ is contained in a package remainder set. So, there is no such $\varphi$, i.e. $\bigcap\gamma^{\prime}(K\cup(K \cro A), (K\cup(K \cro A)) \upVdash n(A^{\prime \prime})) \subseteq K \cro A$.  
\\
\\
\textit{2. From left to right:} We only check \textit{$\cro$-strong Uniformity:} Assume that it holds for all $X \subseteq \mathcal{L}$ that $X \subseteq K_1 \cup (K_1 \cro A)$ and $X \cup ((K_1 \cro A) \setminus K_1) \nvdash \falsum$ iff $X \subseteq K_2 \cup (K_2 \cro B)$ and $X \cup ((K_2 \cro B) \setminus K_2) \nvdash \falsum$. Then, by Observations \ref{observation on upper bound property} and \ref{lemma for external choice revision}, it follows that $(K_1 \meo A) \upVdash n(A^{\prime}) = (K_2 \meo B) \upVdash n(B^{\prime}) $ where $A^{\prime}  = (K \meo A) \setminus K$ and $B^{\prime}  = (K \meo B) \setminus K$. Since $\cro$ is unified, by the definition of $\cro$, it follows that $K_1 \cro A = K_2 \cro B$. 
\\
\\
\textit{From right to left:} Let $\gammap$ be defined in the same way as the $\gammap$ in the proof of the first part. By Observation \ref{Observation that all selection functions are Join-unified}, we only need to prove that $\gammap$ is $\upVdash$-unified. Let $(K_1 \cup (K_1 \cro A)) \upVdash n(\Ap) = (K_2 \cup (K_2 \cro B)) \upVdash n(\Bp) \neq \emptyset$, where $\Ap = (K_1 \cro A) \setminus K_1$ and $\Bp = (K_2 \cro B) \setminus K_2$. By Observation \ref{observation on upper bound property}, it follows that for all $X \subseteq \mathcal{L}$, $X \subseteq K_1 \cup (K_1 \cro A)$ and $X \nvdash n(\Ap)$ iff $X \subseteq K_2 \cup (K_2 \cro B)$ and $X \nvdash n(\Bp)$. Moreover, it follows from $(K_1 \cup (K_1 \cro A)) \upVdash n(\Ap) = (K_2 \cup (K_2 \cro B)) \upVdash n(\Bp) \neq \emptyset$ that both $\Ap$ and $\Bp$ are non-empty. So it follows that for all $X \subseteq \mathcal{L}$, $X \subseteq K_1 \cup (K_1 \cro A)$ and $X \cup ((K_1 \cro A) \setminus K_1) \nvdash \falsum$  iff $X \subseteq K_2 \cup (K_2 \cro B)$ and $X \cup ((K_2 \cro B) \setminus K_2) \nvdash \falsum$. So, by $\cro$-strong Uniformity, $K_1 \cro A = K_2 \cro B$. Thus, by the definition of $\gammap$, it holds that $\gammap$ is $\upVdash$-unified.
\end{proof}

\begin{proof}[Proof for Theorem \ref{representation theorem for internal making up one's mind}:]
\textit{1. From left to right:} $\MU$-inclusion, $\MU$-success and \linebreak $\MU$-consistency follow directly from $\cro$-inclusion, $\cro$-success and \linebreak $\cro$-consistency (Note that in the proof of Theorem \ref{representation theorem for internal choice revision}, the condition that $K$ is consistent is not used in checking these three postulates). Moreover, $\MU$-coincidence follows immediately from Lemma \ref{lemma on interderivability between postulates} and Theorem \ref{representation theorem for internal choice revision}. So we only supply verifications for $\MU$-iteration and $\MU$-relevance. Let $\MU$ be an operation of making up one's mind, i.e. there exist a choice contraction $\cco$ and a consistency-preserving partial expansion $\meo$ such that $K \MU \varphi  = K \cco n(\{\varphi, \neg \varphi\}) \meo \{\varphi, \neg \varphi\}$.
\\
\textit{$\MU$-iteration:} It follows directly from the definitions of $\cco$ and $\meo$ that $K \MU \taut = (K \cco \perp ) \cup \{\taut\}$. So $(K \MU \taut) \cap K = (K \cap (K \cco \perp)) \cup (K \cap \{\taut\} ) = (K \cco \perp) \cup (K \cap \{\taut\} )$. If $\taut \in K$, then $\taut \in K \cco \perp$ by the definition of $\cco$. So, it always holds that $(K \MU \taut) \cap K = (K \cco \perp) \cup (K \cap \{\taut\} ) = K \cco \perp$. So, $((K \MU \taut) \cap K) \MU \varphi = (K \cco \perp) \cco \perp \meo \{\varphi, \neg \varphi\} = (K \cco \perp) \meo \{\varphi, \neg \varphi\} = K \cco n(\{\varphi, \neg \varphi \}) \meo \{\varphi, \neg \varphi\} = K \MU \varphi$, i.e. $\MU$-iteration holds. 
\\
\textit{$\MU$-relevance:} Let $\psi \in K \setminus K \MU \varphi$, then $K \cco n(\{\varphi, \neg \varphi\}) = K \cco \perp  \neq K$, i.e. $K  \angle \perp \neq \emptyset$. So there exists $X \in K  \angle \perp$ such that $K \cco \perp \subseteq X$ and $\psi \notin X$.  As shown in above, it holds that $K \cap (K \MU \taut ) = K \cco \perp$, so it follows from $K \cco \perp \subseteq X$ that $K \cap (K \MU \taut)  \subseteq X$. Moreover, it follows from $X \in K \angle \perp$ and $\psi \in K \setminus X$ that $X \nVdash \perp$ and $X \cup \{\psi\} \Vdash \perp$, i.e. $X \nvdash \falsum$  and $X \cup \{\psi\} \nvdash \falsum$. Thus, $X$ is just the set we are looking for.
\\
\\
\textit{From right to left:} Let us define a selection function $\gamma$ in the following way:\\
(i) $\gamma(K, \mb{Y}) = K$ if $\mb{Y} = \emptyset$.\\
(ii) If $\mb{Y} = K \angle \perp \neq \emptyset$, then $\gamma(K, \mb{Y}) = \{Y \in \mb{Y} \mid K \cap (K \MU \taut) \subseteq Y \}$.\\
(iii) Otherwise, $\gamma(X, \mb{Y}) = \mb{Y}$.
\\
It is obvious that:\\
(1) $\gamma$ is well-defined, i.e. if $\mb{X} = \mb{Y}$, then $\gamma(K, \mb{X}) = \gamma(K, \mb{Y})$, which is immediate from the definition;\\
(2) $\gamma(K, \mb{Y}) = K$ if $\mb{Y} = \emptyset$, which is also immediate from the definition;\\
(3) $\gamma(K, \mb{Y}) \subseteq \mb{Y}$ if $\mb{Y} \neq \emptyset$, which is immediate from the definition again; and\\
(4) $\gamma(K, \mb{Y}) \neq \emptyset$ if $\mb{Y} \neq \emptyset$, which is immediate from that $ K \cap (K \MU \taut)$ is consistency due to $\MU$-consistency.
\\
Furthermore, let us define another selection function $\gammap$ in the following way:\\
(i) $\gamma^{\prime}(X, \mb{Y}) = X$ if $\mb{Y} = \emptyset$.\\
(ii) If there exist set $K$ and sentence $\varphi$ such that $X = K \cap (K \MU \taut)$ and $\mb{Y} = X \Join \{\varphi, \neg \varphi\} $, then $\gamma^{\prime}(X, \mb{Y}) = \{Y \in \mb{Y} \mid Y \subseteq K \MU \varphi \}$.\\
(iii) If $\{Y \in \mb{Y} \mid Y \mbox{ is consistent}\} \neq \emptyset$ and there is no set $K$ and sentence $\varphi$ such that $X = K \cap (K \MU \taut)$ and $\mb{Y} = X \Join \{\varphi, \neg \varphi\} $, then $\gamma^{\prime}(X, \mb{Y}) \in \{Y \in \mb{Y} \mid Y \mbox{ is consistent}\}$.\\
(iv) Otherwise, $\gamma^{\prime}(X, \mb{Y}) = \mb{Y}$.
\\
We need to show that:\\
$(1^{\prime})$ $\gamma^{\prime}$ is well-defined, i.e. if $\mb{X} = \mb{Y}$, then $\gammap (X, \mb{X}) = \gammap (X, \mb{Y})$;\\
$(2^{\prime})$ $\gammap (X, \mb{Y}) = K$ if $\mb{Y} = \emptyset$, which is immediate from the definition;\\
$(3^{\prime})$ $\gammap (X, \mb{Y}) \subseteq \mb{Y}$ if $\mb{Y} \neq \emptyset$, which is also immediate from the definition;\\
$(4^{\prime})$ $\gammap (X, \mb{Y}) \neq \emptyset$ if $\mb{Y} \neq \emptyset$; and\\
$(5^{\prime})$ $\gammap$ is $\Join$-consistency-preserved, i.e. $\bigcup \gammap (X, \mb{Y})$ is consistent, if there exists a set $A$ such that $\mb{Y} = X \Join A$ and $\{Y \in \mb{Y} \mid Y \mbox{ is consistent}\} \neq \emptyset$.\\ 
For $(1^{\prime})$, let $X = K_1 \cap (K_1 \MU \taut) = K_2 \cap (K_2 \MU \taut) $ and $X \Join \{\varphi , \neg \varphi\} = X \Join \{\psi , \neg \psi\} $. We should consider two cases. (i) $X \cap \{\varphi, \neg \varphi\} \neq \emptyset$. Then, it follows that $X \cap \{\psi , \neg \psi\}  \neq \emptyset$. Moreover, it follows from $X \Join \{\varphi , \neg \varphi\} = X \Join \{\psi , \neg \psi\} $ that $X \cup \{\varphi , \neg \varphi\} =  X \cup \{\psi , \neg \psi\} $. By $\MU$-consistency, $X \nvdash \falsum$. So, by $\MU$-coincidence, $X \MU \varphi = X \MU \psi$ and hence, by $\MU$-iteration, $K_1 \MU \varphi = K_2 \MU \psi$. (ii) $X \cap \{\varphi, \neg \varphi\} = \emptyset$. Then, it follows immediately that $\varphi$ is identical to $\psi$. So, by $\MU$-iteration, it follows from $K_1 \cap (K_1 \MU \taut) = K_2 \cap (K_2 \MU \taut)$ that $K_1 \MU \varphi = (K_1 \cap (K_1 \MU \taut)) \MU \varphi = (K_2 \cap (K_2 \MU \taut))\MU \psi = K_2 \MU \psi$. Thus, in each case, $\gammap (X, \mb{X}) = \gammap (X, \mb{Y})$.
\\
For $(4^{\prime})$,  it is enough to show that $K \MU \varphi \in ( (K \cap (K \MU \taut)) \Join \{\varphi, \neg \varphi )\}$. So we should show (a) $K \MU \varphi \subseteq (K \cap (K \MU \taut)) \cup \{\varphi, \neg \varphi\} $ which is immediate from $\MU$-inclusion and $\MU$-iteration; (b) $(K \cap (K \MU \taut)) \subseteq K \MU \varphi$; and (c) $ \{\varphi, \neg \varphi\} \cap (K  \MU \varphi) \neq \emptyset$ which is immediate from $\MU$-success. For (b), by $\MU$-iteration, we only need to show $(K \cap (K \MU \taut)) \subseteq (K \cap (K \MU \taut)) \MU \varphi$. By $\MU$-consistency, $K \cap (K \MU \taut) \nvdash \falsum$. So, by $\MU$-relevance, there is no $\psi \in (K \cap (K \MU \taut)) \setminus (K \cap (K \MU \taut)) \MU \varphi$, i.e. $(K \cap (K \MU \taut)) \subseteq (K \cap (K \MU \taut)) \MU \varphi$.  
\\
As to $(5^{\prime})$, we only need to show that $K \MU \varphi$ is consistent for every $\varphi$, which is confirmed by $\MU$-consistency.
\\
To finish the proof,  we need to show that for every set $K$ and sentence $\varphi$, it holds that
\\
$(\star)$ $K \MU \varphi = \bigcup\gamma^{\prime}(\bigcap \gamma(K, K \angle n(\{\varphi, \neg \varphi\})), \bigcap \gamma(K, K \angle n(\{\varphi, \neg \varphi\})) \Join \{\varphi, \neg \varphi\}))$. 
\\
We first prove that $\bigcap \gamma(K, K \angle n(\{\varphi, \neg \varphi\})) = K \cap (K \MU \taut)$. It is easy to see that $K \angle n(\{\varphi, \neg \varphi\}) = K \angle \perp$. So, the from right to left inclusion direction follows immediately from the clause (ii) of the definition of $\gamma$. For the other inclusion direction, suppose $(\dagger)$ there exists $\psi \notin K \cap (K \MU \taut )$ such that $\psi \in  \bigcap \gamma(K,  K \angle \perp)$.  It follows from the definition of $\gamma$ that $\psi \in  K \setminus K \MU \taut$. So, by $\MU$-relevance, there is $K^{\prime}$ such that $K \cap (K \MU \taut) \subseteq K^{\prime} \subseteq K $, $K^{\prime} \nvdash \falsum$  and $K^{\prime} \cup \{\psi\} \vdash \falsum$ . So, by Observation \ref{observation on partial upper bound property of choice remainder set}, there exists $Y$ such that $ K^{\prime } \subseteq Y \subseteq K$ and $Y \in K \angle \perp$. Also, it follows from $K \cap (K \MU \taut) \subseteq K^{\prime} \subseteq Y$ that $Y \in \gamma(K, K \angle \perp)$. Hence, by $(\dagger)$, $\psi \in Y$. However, it follows from $K^{\prime} \cup \{\psi\} \vdash \falsum$ and $K^{\prime}\cup \{\psi\} \subseteq Y$ that $Y \vdash \falsum$,  which contradicts that $Y \in K \angle \perp$. Thus, there is no such kind of $\psi$, i.e. $\bigcap \gamma(K, K \angle \perp) \subseteq K \cap (K \MU \taut)$.\\
So $(\star)$ holds iff\\
$(\star \star)$ $K \MU \varphi = \bigcup\gamma^{\prime}(K \cap (K \MU \taut), (K \cap (K \MU \taut)) \Join \{\varphi, \neg \varphi\})$ holds. \\
As we have shown in the proof for $(4^{\prime})$, $K \MU \varphi \in (K \cap (K \MU \taut)) \Join \{\varphi, \neg \varphi\} $. So, $(\star \star)$ follows immediately from the definition of $\gammap$.  
\\
\\
\textit{2. From left to right:} We only need to check \textit{$\MU$-redundancy:} Let $Z \equiv \{\perp\}$. It is easy to see that $K \cco n(\phinegphi) = K \cco \{\perp\}$ for every $K$ and $\varphi$. So $K \MU \varphi = K \cco n(\phinegphi )  \meo \phinegphi  = K \cco \{\perp\} \meo \phinegphi  =( K\cup Z ) \cco \{\perp\} \meo \phinegphi  = ( K\cup Z ) \cco \phinegphi \meo \phinegphi = (K \cup Z) \MU \varphi$.
\\
\\
\textit{From right to left:} Let $\gamma$ be constructed in the same way as the $\gamma$ in the proof of the first part. By observation \ref{Observation that all selection functions are Join-unified}, in order to complete the proof, it is enough to show that $\gamma$ is $\angle$-unified. Assume $K_1 \angle \perp = K_2 \angle \perp \neq \emptyset$, we need to prove that $K_1 \cap (K_1 \cro \taut) = K_2 \cap (K_2 \cro \taut)$. By observation \ref{observation on partial upper bound property of choice remainder set}, it holds for every $\varphi \in K_1 \setminus \bigcup (K_1 \angle \perp)$ that $\varphi \Vdash \perp$, i.e. $(K_1 \setminus \bigcup (K_1 \angle \perp)) \equiv \{\perp\}$. Hence, by $\cro$-redundancy, $K_1 \cro \taut = \bigcup (K_1 \angle \perp) \cro \taut$. By $\MU$-consistency, $K_1 \cap (K_1 \cro \taut) \nvdash \falsum$ and hence $K_1 \cap (K_1 \cro \taut) \subseteq \bigcup (K_1 \angle \perp) $ by Observation \ref{observation on partial upper bound property of choice remainder set}. It follows this and $K_1 \cro \taut = \bigcup (K_1 \angle \perp) \cro \taut$ that $K_1 \cap (K_1 \cro \taut) = (\bigcup (K_1 \angle \perp) \cap ((\bigcup (K_1 \angle \perp) \cro \taut)$. By a similar argument, it can be obtained that $K_2 \cap (K_2 \cro \taut) = (\bigcup (K_2 \angle \perp) \cap ((\bigcup (K_2 \angle \perp) \cro \taut)$. Moreover, it follows from $K_1 \angle \perp = K_2 \angle \perp \neq \emptyset$ that $(\bigcup (K_1 \angle \perp) = (\bigcup (K_2 \angle \perp)$. So, $K_1 \cap (K_1 \cro \taut) = K_2 \cap (K_2 \cro \taut)$.
\end{proof}

\begin{proof}[Proof for Theorem \ref{representation theorem for external making up one's mind}:]
\textit{1. From left to right:} It follows immediately from Lemma \ref{lemma on interderivability between postulates} and Theorem \ref{representation theorem for external choice revision}.
\\
\textit{From right to left:} Let us define a selection function $\gamma$ in the following way:\\
(a) $\gamma(K, \mb{Y}) = K$ if $\mb{Y} = \emptyset$.\\
(b) If there exists $\varphi$ such that $\mb{Y} = K \Join \{\varphi, \neg \varphi\}$, then $\gamma(K, \mb{Y}) = \{Y \in \mb{Y} \mid Y \subseteq K \cup (K \MU \varphi) \}$.\\
(c) Otherwise, $\gamma(K, \mb{Y}) = \mb{Y}$.
\\
We need to show that:\\
(1) $\gamma$ is well-defined, i.e. if $\mb{X} = \mb{Y}$, then $\gamma(K, \mb{X}) = \gamma (K, \mb{Y})$;\\
(2) $\gamma(K, \mb{Y}) = K$ if $\mb{Y} = \emptyset$, which is immediate from the definition;\\
(3) $\gamma(K, \mb{Y}) \subseteq \mb{Y}$ if $\mb{Y} \neq \emptyset$, which is also immediate from the definition;\\
(4) $\gamma(K, \mb{Y}) \neq \emptyset$ if $\mb{Y} \neq \emptyset$;\\
For (1) we only need to show that if $K \Join \{\varphi, \neg \varphi\} = K \Join \{\psi, \neg \psi\}$, then $\gamma(K, K \Join \{\varphi, \neg \varphi\}) = \gamma (K, K \Join \{\psi, \neg \psi\})$. If $\phinegphi \cap K = \emptyset$, it is immediate from $K \Join \phinegphi = K \Join \psinegpsi$ that $\varphi = \psi$. So $K \cup (K \MU \varphi) = K \cup (K \MU \psi)$, whence $\gamma(K, K \Join \phinegphi) = \gamma (K, K \Join \psinegpsi)$. On the other hand, if $\phinegphi \cap K \neq \emptyset$, then $K \in K \Join \phinegphi$. So $K \in K \Join \psinegpsi$ and hence $\psinegpsi \cap K \neq \emptyset$. Moreover, it follows from $K \Join \phinegphi = K \Join \psinegpsi$ that $K \cup \phinegphi = K \cup \psinegpsi$. So, by $\cro$-coincidence, it holds that $K \MU \varphi = K \MU \psi$, i.e. $K \cup (K \MU \varphi) = K \cup (K \MU \psi)$, whence $\gamma(K, K \Join \phinegphi) = \gamma (K, K \Join \psinegpsi)$ as well.
\\
For (4), it is enough to show that $K \cup (K \MU \varphi) \in K \Join \phinegphi$. This follows from that (a) $K \subseteq K \cup (K \MU \varphi)$ which is trivial; (b) $K \cup (K \MU \varphi) \subseteq K \cup \phinegphi$ which is immediate from $\MU$-inclusion; and (c) $ \phinegphi \cap (K \cup (K \MU \varphi)) \neq \emptyset$ which is immediate from $\MU$-success.
\\
Furthermore, let us define another selection function $\gamma^{\prime}$ in the following way:\\
(a) $\gamma^{\prime}(X, \mb{Y}) = X$ if $\mb{Y} = \emptyset$.\\
(b) If there exists $K$ and $\varphi$ such that $X = K \cup (K \MU \varphi)$ and $\mb{Y} = X \upVdash n(A^{\prime}) \neq \emptyset$ where $A^{\prime} = (K \MU \varphi) \setminus K$, then $\gamma^{\prime}(X, \mb{Y}) = \{Y \in \mb{Y} \mid K \MU \varphi \subseteq Y \}$.\\
(c) Otherwise, $\gamma^{\prime}(X, \mb{Y}) = \mb{Y}$.
\\
We need to show that:\\
$(1^{\prime})$ $\gamma^{\prime}$ is well-defined, i.e. the case of $\mb{X} = \mb{Y}$, then $\gamma^{\prime}(K, \mb{X}) = \gamma^{\prime} (K, \mb{Y})$;\\
$(2^{\prime})$ $\gamma^{\prime}(K, \mb{Y}) = K$ if $\mb{Y} = \emptyset$, which is immediate from the definition;\\
$(3^{\prime})$ $\gamma^{\prime}(K, \mb{Y}) \subseteq \mb{Y}$ if $\mb{Y} \neq \emptyset$, which is also immediate from the definition;\\
$(4^{\prime})$ $\gamma^{\prime}(K, \mb{Y}) \neq \emptyset$ if $\mb{Y} \neq \emptyset$;\\
For $(1^{\prime})$, we only consider the principle case, i.e. let $X = (K_1 \MU \varphi) \cup K_1  = (K_2 \MU \psi) \cup K_2$ and $ X \upVdash n(A^{\prime}) = X \upVdash n(B^{\prime}) \neq \emptyset $ where $A^{\prime} = (K_1 \MU \varphi) \setminus K_1$ and $B^{\prime} = (K_2 \MU \psi) \setminus K_2$. It follows from $X \upVdash n(A^{\prime}) = X \upVdash n(B^{\prime}) \neq \emptyset $ that  $\Ap \neq \emptyset$ and $\Bp \neq \emptyset$. Moreover, by  Observation \ref{observation on upper bound property}, it follows from $X \upVdash n(A^{\prime}) = X \upVdash n(B^{\prime}) $ that for all $X^{\prime} \subseteq X$, $X^{\prime} \nvdash n(A^{\prime})$ iff $X^{\prime} \nvdash n(B^{\prime})$. So, $K_1 \neq ((K_1 \MU \varphi) \cup K_1 ) =X = ((K_2 \MU \psi) \cup K_2) \neq K_2$ and it holds for all $X^{\prime} \subseteq X$ that $X^{\prime} \cup ((K_1 \MU \varphi) \setminus K_1) \nvdash \falsum$  iff $X^{\prime} \cup  ((K_2 \MU \psi) \setminus K_2) \nvdash \falsum$. Hence, by $\MU$-Uniformity, $K_1 \MU \varphi = K_2 \MU \psi$. Thus, $\gamma^{\prime}(X, X \upVdash A^{\prime}) = \gamma^{\prime}(X, X \upVdash B^{\prime})$.
\\
For $(4^{\prime})$, we only need to show that $\gamma^{\prime}(K \cup (K \MU \varphi), \mb{Y}) \neq \emptyset$ when $\mb{Y} = ((K \cup (K \MU \varphi)) \upVdash n(A^{\prime})) \neq \emptyset$ where $A^{\prime} = (K \MU \varphi) \setminus K$. We first show that $K \MU \varphi \nvdash n(A^{\prime})$. It follows from $((K \cup (K \MU \varphi)) \upVdash n(A^{\prime})) \neq \emptyset$ that $A^{\prime}$ is consistent and non-empty. So, by $\MU$-Consistency, $K \MU \varphi$ is consistent. Moreover, $A^{\prime} = ((K \MU \varphi) \setminus K) \subseteq K \MU \varphi$, i.e. $K \MU \varphi \Vdash A^{\prime}$. Hence, $K \MU \varphi \nvdash n(A^{\prime})$. So, by Observation \ref{observation on upper bound property}, there exists $Y$ such that $K \MU \varphi \subseteq Y \subseteq K\cup (K \MU \varphi)$ and $Y \in \mb{Y}$.\\
To complete the proof, we will show it holds that\\
$(\star)$ $K \MU \varphi = \bigcap\gamma^{\prime}(\bigcup \gamma(K, K \Join \phinegphi), (\bigcup \gamma(K, K \Join \phinegphi)) \upVdash n(A^{\prime}))$, where $A^{\prime} = (\bigcup \gamma(K, K \Join \phinegphi)) \setminus K$. \\
As we have shown in the proof for $(4)$, $K \cup (K \MU \varphi) \in K \Join \phinegphi$, so, by the definition of $\gamma$, $(\star)$ holds iff \\
$(\star \star)$ $K \MU \varphi = \bigcap\gamma^{\prime}(K\cup(K \MU \varphi), (K\cup(K \MU \varphi)) \upVdash n(A^{\prime \prime}))$ holds, where $A^{\prime \prime} = (K \MU \varphi) \setminus K $.\\
As to $(\star \star)$, we need to consider two cases, according to whether $(K\cup(K \MU \varphi)) \upVdash n(A^{\prime \prime}))$ is empty. In the first case, $(K\cup(K \MU \varphi)) \upVdash n(A^{\prime \prime})) = \emptyset$. Then, it is easy to see that either (i) $n(A^{\prime \prime}) = \emptyset$ or (ii)$A^{\prime \prime}$ is inconsistent. Suppose (i), it follows immediately that $K \MU \varphi \subseteq K$. So, by $\MU$-confirmation, it follows that $K \MU \varphi = K$, whence $(\star \star)$ follows immediately. Suppose (ii), it follows immediately that $K \MU \varphi$ is inconsistent. So, by $\MU$-Relevance, $K \subseteq (K \MU \varphi)$, i.e. $K \cup (K \MU \varphi) = K \MU \varphi$. Hence, $(\star \star)$ follows from the definition $\gammap$. 
\\
Let us proceed the proof by considering the remaining case that \linebreak $(K\cup(K \MU \varphi)) \upVdash n(A^{\prime \prime}))$ is non-empty. In this case, the from left to right inclusion direction follows immediately from the clause (b) of the definition of $\gamma^{\prime}$. For the other inclusion direction, let us suppose towards contradiction that $(\dagger)$ there exists $\psi \in (K \cup (K \MU \varphi))$ such that $\psi \notin K \MU \varphi$ but $\psi \bigcap \in \gamma^{\prime}(K\cup(K \MU \varphi), (K\cup(K \MU \varphi)) \upVdash n(A^{\prime \prime}))$. It follows from $\varphi \in (K \cup (K \MU \varphi))$ and $\varphi \notin K \MU \varphi$  that $\varphi \in K \setminus K \MU \varphi$. So, by $\MU$-Relevance, there is some $K^{\prime}$ with $K \MU \varphi \subseteq K^{\prime} \subseteq K \cup (K \MU \varphi)$ such that $K^{\prime} \nvdash \falsum$ and $K^{\prime} \cup \{\psi\} \vdash \falsum$. It follows from $K \MU \varphi \subseteq K^{\prime}$ that $K^{\prime} \Vdash A^{\prime \prime}$. Since $A^{\prime \prime} \neq \emptyset$ and $K^{\prime} $ is consistent, it follows that $K^{\prime} \nvdash n(A^{\prime \prime})$. So, by Observation \ref{observation on upper bound property}, there exists $X$ with $K \MU \varphi \subseteq K^{\prime} \subseteq  X $ such that  $X \in (K\cup(K \cro A)) \upVdash n(A^{\prime \prime})$, and hence $X \in \gamma^{\prime}(K\cup(K \cro A), (K\cup(K \cro A)) \upVdash n(A^{\prime \prime}))$. So, by $(\dagger)$, $\psi \in X$. However, it follows from $K^{\prime} \subseteq X$ and $K^{\prime} \cup \{\psi\} \vdash \perp$ that $X \cup \{\psi\} \vdash \perp$, i.e. $X \vdash \falsum$, which contradicts that $X$ belongs to a package remainder set. So, there is no such $\psi$, i.e. $\bigcap\gamma^{\prime}(K\cup(K \MU \varphi), (K\cup(K \MU \varphi)) \upVdash n(A^{\prime \prime})) \subseteq K \MU \varphi$.  
\\
\\
\textit{2. From left to right:} We only check \textit{$\MU$-strong Uniformity:} Assume it holds for all $X \subseteq \mathcal{L}$ that $X \subseteq K_1 \cup (K_1 \MU \varphi)$ and $X \cup ((K_1 \MU \varphi) \setminus K_1) \nvdash \falsum$  iff $X \subseteq K_2 \cup (K_2 \MU \psi)$ and $X \cup ((K_2 \MU \psi) \setminus K_2) \nvdash \falsum$. Then, by Observations \ref{observation on upper bound property} and  \ref{lemma for external choice revision}, it follows that $(K_1 \meo \phinegphi ) \upVdash n(A^{\prime}) = (K_2 \meo \psinegpsi ) \upVdash n(B^{\prime}) $ where $A^{\prime}  = (K \meo \phinegphi) \setminus K$ and $B^{\prime}  = (K \meo \psinegpsi) \setminus K$. Since $\MU$ is unified, by the definition of $\MU$, it follows that $K_1 \MU \varphi = K_2 \MU \psi$. 
\\
\\
\textit{From right to left:} Let $\gammap$ be defined in the same way as the $\gammap$ in the proof of the first part. To finish the proof, by Observation \ref{Observation that all selection functions are Join-unified}, we only need to prove that $\gammap$ is $\upVdash$-unified. Let $(K_1 \cup (K_1 \MU \varphi)) \upVdash n(\Ap) = (K_2 \cup (K_2 \MU \psi)) \upVdash n(\Bp) \neq \emptyset$, where $\Ap = (K_1 \MU \varphi) \setminus K_1$ and $\Bp = (K_2 \MU \psi) \setminus K_2$. By Observation \ref{observation on upper bound property}, it follows that it holds for all $X \subseteq \mathcal{L}$, $X \subseteq K_1 \cup (K_1 \MU \varphi)$ and $X \nvdash n(\Ap)$ iff $X \subseteq K_2 \cup (K_2 \MU \psi)$ and $X \nvdash n(\Bp)$. Moreover, it follows from $(K_1 \cup (K_1 \MU \varphi)) \upVdash n(\Ap) = (K_2 \cup (K_2 \MU \psi)) \upVdash n(\Bp) \neq \emptyset$ that both $\Ap$ and $\Bp$ are non-empty. So it follows that for all $X \subseteq \mathcal{L}$, $X \subseteq K_1 \cup (K_1 \MU \varphi)$ and $X \cup ((K_1 \MU \varphi) \setminus K_1) \nvdash \falsum$ iff $X \subseteq K_2 \cup (K_2 \MU \psi)$ and $X \cup ((K_2 \MU \psi) \setminus K_2) \nvdash \falsum$. So, by $\cro$-strong Uniformity, $K_1 \MU \varphi = K_2 \MU \psi$. Thus, by the definition of $\gammap$, it holds that $\gammap$ is $\upVdash$-unified.
\end{proof}

\bibliography{kappa}
\bibliographystyle{apalike}

\end{document}